\documentclass[11pt,letter]{article}

\usepackage{fullpage}
\usepackage[utf8]{inputenc}
\usepackage{amsmath}
\usepackage{amsfonts}
\usepackage{amssymb}
\usepackage{comment}
\newcommand{\lap}{\mathcal{L}}
\newcommand{\mlap}{\lap}

\newcommand{\Rset}{\mathbb{R}}

\newcommand{\phideriv}{\mathbf{\Psi}}
\newcommand{\pseudo}{{\dagger}}

\newcommand{\otilde}{\widetilde{O}}
\newcommand{\constb}{k}

\newcommand{\oracle}{\mathtt{MOracle}}
\usepackage{algpseudocode}
\newcommand{\calM}{\mathcal{M}}

\usepackage{amsfonts}

\usepackage{xspace}
\usepackage{tabularx}
\usepackage{pstricks}
\usepackage{setspace}
\usepackage{algorithm}

\usepackage{rotating}
\usepackage{minitoc}
\usepackage{enumerate}

\usepackage{color}
\usepackage[colorlinks=true,citecolor=blue,linkcolor=red
]{hyperref}

\usepackage[normalem]{ulem}
\usepackage{fancybox}
\newenvironment{fminipage}%
  {\begin{Sbox}\begin{minipage}}%
  {\end{minipage}\end{Sbox}\fbox{\TheSbox}}

\usepackage{dsfont}

\newcommand{\lowerb}{\gamma}

\newcommand{\given}{\lowerb \mi \preceq \sum_{i=1}^{d}w_i \mm_{i} \preceq \mi}
\newcommand{\out}{(1- O( \epsilon) ) \lowerb \mi \preceq \sum_{i=1}^{d}w'_i \mm_{i} \preceq  \mi}

\newcommand{\rt}{\tilde{O}\left(\frac{1}{\gamma^2\error^{O(1)}}\left(\runtime_{\mb} + \runtime_{\mb^{-1}} + \runtime_{MV} + \runtime_{QF}\right) \right)}
\newcommand{\trt}{\tilde{O}\left(\frac{1}{(\epsilon\cdot \lowerb)^{O(1)}}\left(\runtime_{\mb} + \runtime_{\mb^{-1}} + \runtime_{MV} + \runtime_{QF}\right)\right)}

\newcommand{\fnf}{\mathbf{f}}

\newcommand{\speed}{\frac{1}{2}}

\newcommand{\glap}{\mathcal{L}_G}

\newcommand{\remove}[1]{}

\newcommand{\R}{\mathbb{R}}

\newcommand{\E}{\mathbb{E}}
\newcommand{\rot}{\intercal}

\newcommand{\tr}{\mathrm{tr}}

\newcommand{\cp}{\mathbf{C}_{+}}
\newcommand{\cn}{\mathbf{C}_{-}}
\newcommand{\cs}{\mathbf{C}}

\usepackage{amsthm}

\renewcommand{\leq}{\leqslant}
\renewcommand{\geq}{\geqslant}
\renewcommand{\le}{\leqslant}

\renewcommand{\tilde}{\widetilde}
\renewcommand{\epsilon}{\varepsilon}

\newcommand{\nnz}{\text{nnz}}

\newcommand{\mvar}[1]{\textbf{#1}}

\newcommand{\mi}{\mvar{I}}
\newcommand{\ma}{\mvar{A}}
\newcommand{\mb}{\mvar{B}}
\newcommand{\mc}{\mvar{C}}
\newcommand{\md}{\mathbf{D}}

\newcommand{\mzero}{\mathbf{0}}
\newcommand{\mm}{\mvar{M}}
\newcommand{\mDelta}{\mathbf{\Delta}}
\newcommand{\mPi}{\mathbf{\Pi}}

\usepackage{setspace}

\usepackage{aliascnt}

\newtheorem{thm}{Theorem}[section]  
\newtheorem{theorem}[thm]{Theorem}

\renewcommand{\tilde}{\widetilde}
\usepackage{todonotes}

\newaliascnt{lemma}{thm}
\newtheorem{lemma}[lemma]{Lemma}
\aliascntresetthe{lemma}

\newaliascnt{lem}{thm}
\newtheorem{lem}[lemma]{Lemma}
\aliascntresetthe{lem}

\newaliascnt{defi}{thm}
\newtheorem{defi}[defi]{Definition}
\aliascntresetthe{defi}

\newaliascnt{prob}{thm}

\aliascntresetthe{prob}

\def\equationautorefname~#1\null{Equation~(#1)\null}	

\usepackage{thmtools}
\usepackage{thm-restate}
\usepackage{thm-autoref}

\numberwithin{equation}{section}

\renewcommand{\vec}[1]{#1}
\usepackage{bm}

\usepackage{microtype}

\newcommand{\runtime}{\mathcal{T}}

\newcommand{\defeq}{\stackrel{\scriptscriptstyle{\text{def}}}{=}}
\newcommand{\onesVec}{\vec{1}}
\newcommand{\zeroVec}{\vec{0}}
\newcommand{\vones}{\onesVec}

\newcommand{\MTPtwo}{{\mathrm{MTP}_2}}

\newcommand{\error}{\rho}

\newcommand{\mH}{\mvar{H}}
\newcommand{\eye}{\mvar{I}}
\newcommand{\vvar}[1]{\vec{#1}}
\newcommand{\vb}{\vvar{b}}
\newcommand{\vx}{\vvar{x}}
\newcommand{\ve}{\vvar{e}}
\newcommand{\vzeros}{\vvar{0}}
\newcommand{\vv}{\vvar{v}}
\newcommand{\mx}{\mvar{X}}

\begin{document}

\title{Efficient Structured Matrix Recovery and
Nearly-Linear Time Algorithms for Solving Inverse Symmetric $M$-Matrices
}

\author{Arun Jambulapati \\
Stanford University\\
jmblpati@stanford.edu
\and
Kirankumar Shiragur \\
Stanford University\\
shiragur@stanford.edu
\and 
Aaron Sidford \\
Stanford University\\
sidford@stanford.edu
}

\date{}

\maketitle

\begin{abstract}
In this paper we show how to recover a spectral approximations to broad classes of structured matrices using only a polylogarithmic number of adaptive linear measurements to either the matrix or its inverse. Leveraging this result we obtain faster algorithms for variety of linear algebraic problems. Key results include:
\begin{itemize}
\item A nearly linear time algorithm for solving the inverse of symmetric $M$-matrices, a strict superset of Laplacians and SDD matrices.
\item An $\tilde{O}(n^2)$ time algorithm for solving $n \times n$ linear systems that are constant spectral approximations of Laplacians or more generally, SDD matrices.
\item An $\tilde{O}(n^2)$ algorithm to recover a spectral approximation of a $n$-vertex graph using only $\tilde{O}(1)$ matrix-vector multiplies with its Laplacian matrix.
\end{itemize}
The previous best results for each problem either used a trivial number of queries to exactly recover the matrix or a trivial $O(n^\omega)$ running time, where $\omega$ is the matrix multiplication constant.

We achieve these results by generalizing recent semidefinite programming based linear sized sparsifier results of Lee and Sun (2017) and providing iterative methods inspired by the semistreaming sparsification results of Kapralov, Lee, Musco, Musco and Sidford (2014) and input sparsity time linear system solving results of Li, Miller, and Peng (2013). 
We hope that by initiating study of these natural problems, expanding the robustness and scope of recent nearly linear time linear system solving research, and providing general matrix recovery machinery this work may serve as a stepping stone for faster algorithms. 

\end{abstract}

\newpage

\section{Introduction}

Given a $n$ vertex undirected graph $G$ with non-negative edge weights its Laplacian $\lap \in \R^{n \times n}$ is defined as $\mlap = \md - \ma$ where $\md$ is its weighted degree matrix and $\ma$ its weighted adjacency matrix. This matrix is fundamental for modeling large graphs and the problem of solving Laplacian systems, i.e. $\lap x = b$, encompasses a wide breadth of problems including computing electric current in a resistor network, simulating random walks on undirected graphs, and projecting onto the space of circulations in a graph (see \cite{Teng10, Vishnoi13} for surveys). Correspondingly, the problem of solving Laplacian systems is incredibly well-studied and in a celebrated result of Spielman and Teng in 2004 \cite{SpielmanT04} it was shown that Laplacian systems can be solved in nearly linear time. 

Over the past decade nearly linear time Laplacian system solving has emerged as an incredibly powerful hammer for improving the asymptotic running time of solving a wide variety of problems \cite{DaitchS08,KyngRSS15,CohenMSV17,CohenMTV17}. The fastest known algorithms for solving a variety of problems including, maximum flow \cite{ChristianoKMST11,LeeRS13,KLOS14,LS14,Madry16}, sampling random spanning trees \cite{KelnerM09,DurfeeKPRS17,Schild18}, and graph clustering \cite{SpielmanT04,OrecchiaV11,OrecchiaSV12} all use the ability to solve Laplacian systems in nearly linear time. Moreover, beyond the direct black-box use of Laplacian systems solvers, there has been extensive work on developing new algorithms for provably solving Laplacian systems in nearly linear time \cite{KoutisMP10,KoutisMP11,KOSZ13,LS13,CohenKMPPRX14,PengS14,KyngLPSS16,KyngS16}, many of which have had broad algorithmic implications. 

However, despite the prevalence of Laplacians systems and the numerous algorithmic approaches for solving them the class of linear systems solvable by this approach is in some ways brittle. Though there have been interesting extensions to solving block diagonally dominant systems \cite{KyngLPSS16}, M-matrices \cite{DaitchS08,ajss18}, and directed Laplacians \cite{CohenKPPSV16,CohenKPPRSV17} there are still simple classes of matrices closely related to Laplacians for which our best running times are achieved by ignoring the graphical structure of the problem and using fast matrix multiplication (FMM) black box. 

\paragraph{Laplacian Pseudoinverses:} One natural instance of this is solving linear systems in the \emph{Laplacian psuedoinverse}. Here we are given a matrix $\mm \in \R^{n \times n}$ which we are promised is the psuedoinverse of a Laplacian, i.e. $\mm = \lap^\dagger$ for some Laplacian $\mlap \in \R^{n \times n}$, and wish to solve the linear system $\mm x = b$ or recover a spectral sparsifier of the associated graph. This problem arises when trying to fit a graph to data or recover a graph from effective resistances, a natural distance measure (see \cite{effResRecArx18} for motivation and discussion of related problems.) More broadly, the problem of solving linear systems in inverse symmetric M-matrices is prevalent and corresponds to statistical inference problems involving distributions that are multivariate totally positive of order 2 ($\MTPtwo$) \cite{KARLIN1983419,mm14,fallat2017}.

\paragraph{Perturbed Laplacians:} Another natural instance of this is solving linear systems in spectral approximations of Laplacians or small perturbations of Laplacians. Here we are given a matrix $\mm \in \R^{n \times n}$ which we are promised is close spectrally to a Laplacian, i.e. $\gamma\mlap \preceq \mm \preceq  \mlap$ for Laplacian $\mlap \in \R^{n \times n}$ and $\gamma>0$, and wish to solve the linear system $\mm x = b$. Such systems could occur naturally in numerical settings where Laplacian matrices are used as approximations of physical phenomena \cite{BomanHV08}.\\
\\
Despite the abundant graph structure in these problems and wealth of machinery for provably solving Laplacian systems, the best known running times for each of these problems are achieved by ignoring this structure and solving them as arbitrary linear systems using FMM, yielding a naive running time of $O(n^\omega)$ where $\omega < 2.373$ is the matrix multiplication constant \cite{Williams12}. 

In this paper we initiate a theoretical investigation into these natural linear system solving problems and provide improved $\tilde{O}(n^2)$ time algorithms for solving inverse M-matrices, Laplacian pseudo-inverses and perturbed Laplacians (when $\gamma \in \Omega(1/\mathrm{poly}(\log n))$). Consequently, our algorithms run in nearly linear time whenever the input matrix is dense. This is typically the case for Laplacian psuedoinverses and can easily happen for perturbed Laplacians. Furthermore, we show that the inverse of symmetric M-matrices are either block separable or dense  and therefore we can solve them in nearly linear time (see \autoref{sec:inv_dense}). 

To obtain these results we provide more general results on \emph{matrix recovery} that we believe are surprising and of independent interest. For example, we prove that given an oracle which can make black-box matrix vector queries to a matrix $\mm \in \R^{n \times n}$ such that $\mm = \sum_{i \in [d]} \alpha_i \mm_i$ where each $\alpha_i \geq 0$ and $\mm_i$ is symmetric positive semidefinite (PSD) we can recover a spectral approximation to $\mm$ using only $\otilde(1)$ queries. 
This result is analogous (but not quite comparable) to a result in \cite{KapralovLMMS14} providing streaming spectral sparsifiers. 
The result is perhaps surprising as the more general problem of recovering an arbitrary PSD matrix from $\otilde(1)$ linear measurements is impossible \cite{AndoniCKQWZ16}.

We achieve our results through a careful combination and adaptation of several previous algorithmic frameworks sparsification and solving linear systems. In particular we show how to extend the semidefinite programming and sparsification framework of Lee and Sun \cite{ls17} and apply it iteratively through linear system solving machinery introduced in the work of Li, Miller, and Peng \cite{LiMP13} similar to how it was applied in Kapralov, Lee, Musco, Musco and Sidford \cite{KapralovLMMS14}. While further insights are needed to adapt this framework for solving Laplacian psuedoinverses and inverse M-matrices, we believe the primary strength of this work is demonstrating that these natural (but overlooked) linear system solving problems can be solved easily and by providing general and powerful frameworks to solve them.

We remark that in \cite{CG18} there was also a result similar to \autoref{thm:recmain}. This result too followed from a careful interpretation of \cite{ls17} and was done in the case of graphs for a different applications. Our \autoref{thm:recmain} slightly generalizes both \cite{ls17, CG18} in our ability to handle arbitrary $\gamma > 0$ and emphasis on working in restricted oracle models.

Our work is a crucial step towards answer the following fundamental question, \emph{when is applying a matrix to a vector no harder (up to polylogarithmic factors) than solving a linear system in that matrix or approximately recovering its quadratic form}? Our work provides a positive answers in terms of query complexity and running time for a broad class of systems. We hope this work can serve as a stepping stone to further expand the scope of provably efficient linear systems solving.

\paragraph{Paper Organization} 
In \autoref{sec:prelim} we cover preliminaries and notation used through the paper. In \autoref{sec:results} we provide a formal statement of the main results of this paper and compare to previous work. In \autoref{sec:solvers} we 
present the proof of our main results given the generalization of \cite{ls17} and in \autoref{sec:recovery} we prove the generalization of \cite{ls17}. In \autoref{sec:square_root} we provide details about approximating matrix square roots we use in our algorithm. In \autoref{sec:JL}
 we prove additional details about how to efficiently implement the algorithm in \autoref{sec:recovery} and in \autoref{sec:inv_dense} we provide additional details about how to solve symmetric inverse M-matrices in nearly linear time.

\newcommand{\kernel}{\mathrm{ker}}
\newcommand{\vecx}{x}
\newcommand{\MM}{\mm}
\newcommand{\M}[1]{\textbf{#1}}
\newcommand{\OO}[1]{\otilde(#1)}
\newcommand{\Half}{\frac{1}{2}}
\section{Preliminaries}
\label{sec:prelim}

Here we provide an overview of the notation and basic mathematical facts we use throughout. Given our extensive use machinery from \cite{ls17} we choose similar notation.

\subsection{Notation}

\textbf{Basics}: We let $[n] \defeq {1,...,n}$, $\onesVec$ denote the all ones vector, $\zeroVec$ denote the all zeros vector, $\mi$ to denote the identity matrix, and $\mzero$ to denote the all zeros matrix. We let $\nnz(\ma)$, $\nnz(v)$ denote number of non-zero entries of matrix $\ma$ and vector $v$.
We use $\otilde(\cdot)$ to hide terms logarithmic in $n, m, \epsilon$ and condition number of PSD matrices. 

\textbf{Structured Matrices}: We call a matrix $\ma$ a \emph{Z-matrix} if $\ma_{ij} \leq 0$ for all $i \neq j$. We call a matrix $\mlap$ a \emph{Laplacian} if it is a symmetric Z-matrix with $\mlap \onesVec = \zeroVec$. We call a matrix $\ma$ \emph{diagonally dominant (DD)} if $\ma_{ii} \geq \sum_{j \neq i} \ma_{ij}$ for all $i$ and \emph{symmetric diagonally dominant (SDD)} if it is symmetric and DD. Further, we call $\mm$ an invertible $M$-matrix if $\mm = s \mi- \ma$ where $s > 0$, $\ma \in \R^{n \times n}_{\geq 0}$ and $\rho(\ma) < s$ where $\rho(\ma)$ is the spectral radius of $\ma$.

\textbf{Graph Laplacians}: For an undirected graph $G = (V, E)$ with non-negative edge weights $w \in \R^E_{> 0}$ its Laplacian $\lap \in \R^{V \times V}$ is defined for all $i,j \in V$ by $\lap_{i,j} = -w_{i,j}$ if $i \neq j$ and $\{i, j\} \in E$, $\lap_{i, j} = 0$ if $i \neq j$ and $\{i, j\} \notin E$, and $\lap_{i,j} = \sum_{k \neq i} \lap_{i,k}$ if $i = j$. Further we let $\mlap_{K_n} \defeq n \mi - \vones \vones^\top$ denote the Laplacian of the unit weight complete graph. Note that all Laplacians are SDD

\textbf{Spectrum}: For symmetric matrix $\ma \in \R^{n \times  n}$ we let $\lambda_1(\ma) \leq \lambda_2(\ma) \leq ... \leq \lambda_{n}(\ma)$ denote the eigenvalues of $\ma$ and let $\lambda_{\min} \defeq \min_{i \in [n], \lambda_{i} \neq 0} \lambda_i(\ma)$ and $\lambda_{\max} \defeq \lambda_n(\ma)$. Further, we let $\kappa(\ma) \defeq \lambda_{\max}(\ma) / \lambda_{\min}(\ma)$ denote the condition number of $\ma$ and $\kernel(\ma)$ denote the kernel of $\ma$.

\textbf{Positive Semidefiniteness}: We call a symmetric matrix $\ma \in \R^{n \times n}$ positive semi-definite (PSD) if and only if $x^\top \ma x \geq 0$ for all $ x \in \R^{n}$ or equivalently all the eigenvalues of $\ma$ $\lambda_i(\ma) \geq 0 $ for all $i \in [n]$. We use notation $\ma \succeq 0$ to say matrix $\ma$ is PSD, $\ma \succeq \mb$ to denote the condition $\ma-\mb \succeq \mzero$, and define $\preceq$, $\succ$, and $\prec$ analogously. We call a symmetric matrix $\ma$ positive definite (PD) iff $\ma \succ 0$. We further use $\ma \approx_{\epsilon} \mb$ to denote $(1-\epsilon) \ma \preceq \mb \preceq (1+\epsilon) \ma$. We call a \emph{Z-matrix} $\ma$ SDDM if it is DD and PD.

\textbf{Functions of PSD Matrices}: For any matrix $\ma \succeq 0$ and a positive real valued function $\fnf$, we use $\fnf(\ma)$ to denote the matrix with eigenvalues $\fnf(\lambda_i(\ma))$ for all $i \in [n]$ and eigenvectors unchanged. For a PSD matrix $\ma$, its pseudo-inverse, denoted  $\ma^{\dagger}$, is defined as the matrix with same kernel and non-zero eigenvalues inverted while keeping all  eigenvectors unchanged. 

\textbf{Norm}: For any vector $\vecx \in \R^{n}$, matrix $\ma \in \R^{n \times n}$ and a symmetric PSD matrix $\mH \in \R^{n \times n}$, we define $\|\vecx\|_{\mH}\defeq\sqrt{\vecx^{\top}\mH \vecx}$ and $\|\ma\|_{\mH}\defeq\max_{\vecx \neq \zeroVec} \frac{\|\ma\vecx\|_{\mH}}{\|\vecx\|_{\mH}}$. Note that $\|\ma\|_{\mH} = \|\mH^{1/2} \ma \mH^{-1/2}\|_2$

$(1 - \epsilon)$-\textbf{Spectral Approximation}: We say a matrix $\ma$ is \emph{$(1-\epsilon)$-spectral approximation} to matrix $\mb$ if and only if $(1-\epsilon)\mb \preceq \ma \preceq (1+\epsilon)\mb $. We use notation $\ma \approx_{\epsilon} \mb$ to denote that  $\ma$ is $\epsilon$-spectral approximation to $\mb$.

\textbf{Matrix Operations}: For matrices $\ma$ and $\mb$ we let $\ma \bullet \mb \defeq \tr(\mb^\top \ma) = \sum_{i,j} \ma_{ij} \mb_{ij}$. Often we use $\runtime_{\ma}$ denote the time required to compute $\ma x$ for any vector $x$.

\subsection{Linear System Solving}

In this paper we make extensive use of preconditioned iterative methods to solve linear systems. In particular we use the following theorem which is a special case of Lemma 4.2 in \cite{CohenKPPRSV17}.

\begin{thm}[Preconditioned Richardson \cite{CohenKPPRSV17}]
\label{thm:precon_rich}
Let $\ma, \mb, \mH \in \R^{n \times n}$ be symmetric PSD matrices such that $\kernel(\mH) \subseteq \kernel(\ma)=\kernel(\mb)$ and let $c \defeq  \|\eye - \eta \mb^{-1} \ma\|_{\mH}$ for $\eta > 0$.  Then  for all $b \in \R^{n}$ and $t > 0$ the point $x_t$ in \autoref{alg:precon_rich} satisfies 
$
\| x_t - \ma^\pseudo b\|_{\mH} \leq c^t \|\ma^\pseudo b\|_{\mH}    
$.

In particular if $\frac{1}{r}\ma \preceq \mb \preceq  \ma$ and $\kappa=\kappa(\mb)$ then \autoref{alg:precon_rich}  yields 
$
  \|x_t - \ma^\pseudo b\|_2 \leq \epsilon \|\ma^\pseudo b\|_2
$
for  $t = \Omega(r \log(\kappa/\epsilon))$ for  $\eta = 1/r$ and $\mH=\mb$. 
\end{thm}

\begin{algorithm}
\begin{algorithmic}[1]
\caption{$\mathtt{PrecondRich}$($\ma \in \R^{n \times n}$, $b \in \R^n$, $\mb$, $\eta > 0$)}
    \State $x_0 = 0$;
    \State $i = 0$;
    \While{not converged}
    \State $x_{i+1} = x_i - \eta \mb^{-1}[\ma x_i - b]$;
    \State $i = i+1$
    \EndWhile
\label{alg:precon_rich}
\end{algorithmic}
\end{algorithm}

This analysis of Preconditioned Richardson Iteration (\autoref{thm:precon_rich}) yields that if we want to solve linear systems in $\ma$ but can only efficiently solve linear systems in $\mb$ that is spectrally close to $\ma$, then we can solve linear systems in $\ma$ to any desired precision $\epsilon$ by performing only $O(\log(1/\epsilon))$ many matrix vector products with $\ma$ and linear system solves on $\mb$. To simplify our narrative and proofs throughout this paper, we assume that all such linearly convergent algorithms return an exact solution in $\otilde(1)$ such operations. Further, if the cost of each iteration of such a method is $\runtime$ we will simply write that we can solve linear systems in $\ma$ in $\otilde(\runtime)$. This assumption is standard and can be easily removed at the cost of logarithmic factors by standard (albeit tedious) techniques.

\subsection{Matrix Square Roots}

In this paper we also make use of a slightly more general version of a result for SDD matrices from \cite{CCLPT14} which helps to prove our \autoref{thm:recovery2}. The proof for this lemma is deferred to \autoref{sec:square_root} and is exactly the same as in \cite{CCLPT14} with slightly generalization. 

\begin{restatable}[Approximating Square Root]{lemma}{lemprecond}
\label{lem:precond}
For any symmetric PD matrix $\MM$, any square matrix $\M{Z}$ such that $\alpha \M{Z} \M{Z}^{\top} \preceq \MM^{-1} \preceq \M{Z} \M{Z}^{\top}$ for $\alpha \in [0,1]$, and any error tolerance $\epsilon > 0$, there exists a linear operator $\tilde{\M{C}}$ which is an $O{(\frac{1}{\alpha}\log(1/(\epsilon\alpha))}$-degree  polynomial of $\MM$ and $\M{Z}$, such that $ \MM^{-1} \approx_{\epsilon} \tilde{\M{C}} \tilde{\M{C}}^{\top}$.
\end{restatable}
\newcommand{\edit}{\gamma}
\renewcommand{\given}{(1-\epsilon)\mi \preceq \sum_{i=1}^{d}w_i \mm_{i} \preceq \mi}
\renewcommand{\out}{(1- O( \epsilon) ) \mi \preceq \sum_{i=1}^{d}w'_i \mm_{i} \preceq  \mi}
\renewcommand{\rt}{\tilde{O}\left(\epsilon^{-O(1)}\left(\runtime_{\mb} + \runtime_{\mb^{-1/2}} + \runtime_{MV} + \runtime_{QF}\right) \right)}
\renewcommand{\trt}{\tilde{O}\left((\edit\epsilon)^{-O(1)}\left(\runtime_{\mb} + \runtime_{\mb^{-1}}+ \runtime_{SQ} + \runtime_{MV} + \runtime_{QF}\right)\right)}

\section{Overview of Results}
\label{sec:results}

We achieve the results of our paper by reducing each problem, e.g. solving Laplacian pseudo-inverse, inverse $M$-matrices, perturbed Laplacians, etc. to solving the following general and fundamental problem about recovering spectral approximations to PSD matrices.

\begin{restatable}[\textbf{Spectral Approximation Problem}]{prob}{probrecovery}
\label{prob:recovery} 
Let $\mathcal{M}=\{\mm_{i}\}_{i=1}^{d}$ be a set of PSD matrices such that there exists $w \in \R^d_{\geq 0}$ with $\edit\mb \preceq \sum_{i \in [d]} w_i \mm_{i} \preceq \mb$ and $\mb$ is a PD matrix. Suppose for any vector $x$ and vector $\alpha \in \R^{d}$ we can only access these parameters performing the following operations:
\begin{itemize}
\item Computing $\mb x$ in time $\runtime_{\mb}$
\item Computing $\mb^{-1} x$ in time $\runtime_{\mb^{-1}}$
\item Computing $\mc x$ for matrix $\mc$ such that $\mc\mc^{\top}=\mb$ in time $\runtime_{SQ}$
\item Computing $\sum_{i \in [d]} \alpha_i \mm_i x$ in time  $\runtime_{MV}$
\item Computing $x^\top \mm_i x$ for all $i \in [d]$ in time $\runtime_{QF}$.
\end{itemize} 
Given this restricted access to the input this problem asks to compute $w' \in \R^{d}_{\geq 0}$ such that:
$$(1- O( \epsilon))\edit \mb \preceq \sum_{i \in [d]} w'_i \mm_{i} \preceq  \mb$$ 
\end{restatable}

To solve the above problem, we show that the sparsification algorithm of \cite{ls17} carefully modified can be shown to solve this problem surprisingly efficiently. In particular, in \autoref{sec:recovery} we show that the insights from \cite{ls17} expanded and adapted suffice to prove the following theorem.

\begin{restatable}[\textbf{Spectral Approximation}]{thm}{thmrecmain}
\label{thm:recmain}
For any $\epsilon \in (0, 1/20)$ and $\gamma > 0$, the Spectral Approximation Problem (\autoref{prob:recovery}) can be solved in time 
\[
\trt ~.
\]
\end{restatable}

Leveraging this result, in \autoref{sec:solvers} we show surprisingly that we can solve the \emph{Spectral Approximation Problem} (\autoref{prob:recovery}) using only matrix-vector product access to $\mb$ (and not $\mb^{-1/2}$).  
Although it may seem that, no longer using the ability to compute matrix vector products with $\mb^{-1/2}$ may imply more measurements with respect to $\mb$ are required, we show we only need a logarithmic more measurements which depends upon the ratio of $\mb$ to some crude given $\sum_i \beta_i \mm_i$. Further our running time increases only increases by a logarithmic factor provided we can efficiently solve linear systems and compute square roots of (known) linear combinations of $\{\mm_i\}_{i=1}^{d}$. Formally we show the following theorem.

\begin{restatable}[\textbf{Matrix Recovery Theorem}]{thm}{thmmatrecmain}
\label{thm:recovery2}
For any $\epsilon \in (0, 1/20)$, the Spectral Approximation Problem (\autoref{prob:recovery}) with additional access to: 
\begin{itemize}
\item An algorithm that computes $(\sum_{i \in [d]} \alpha_i \mm_i)^{-1} x$ for any  $\alpha \in \R_{\geq 0}^{d}$ and $x \in \R^{n}$ in time $\runtime_{solve}$ 
\item An algorithm that computes $\mc x$ for matrix $\mc$ such that $\mc \mc^{T}=\sum_{i \in [d]} \alpha_i \mm_i$ for any $\alpha \in \R_{\geq 0}^{d}$ that satisfies $\sum_{i \in [d]} \alpha_i \mm_i \succ 0$ and $x \in \R^{n}$ in time $\runtime_{root}$.
\item Vector $\beta \in \R_{\geq 0}^{d}$ such that $\frac{1}{\lambda} \edit \mb \preceq \sum_i \beta_i \mm_i \preceq \frac{1}{\mu} \mb$.  
\end{itemize}
can be solved in time
$$ \tilde{O}\left( (\epsilon \edit)^{-O(1)} ( \runtime_\mb + \runtime_{QF} + \runtime_{MV} +\runtime_{solve}+\runtime_{root}) \right)~.$$
\end{restatable} 

As a corollary of this result, we show that we can solve linear systems and recovery problem in matrices that are spectrally close to Laplacians, i.e. perturbed Laplacians, in $\otilde (\edit^{-O(1)}n^2)$ time.

\begin{restatable}[\textbf{Perturbed Laplacian Solver}]{cor}{corperturbedsolver}
\label{cor:perturbedsolver}
Let $\ma \succeq \mzero \in \mathbb{R}^{n \times n}$ be a matrix such that there exists some (unknown) Laplacian $\lap$ where $\edit \ma \preceq \lap \preceq \ma$. Then for any $\epsilon \in (0, 1/20)$ there exists an algorithm that can solve linear systems in $\ma$ in $\otilde(\edit^{-O(1)} n^2)$ time.
 Further if $\lap$ is the Laplacian of a connected graph then there is an algorithm that can recover a Laplacian $\lap'$ such that $(1-O(\epsilon))\edit \ma \preceq \lap' \preceq \ma$ in 
$\otilde((\epsilon \edit)^{-O(1)} n^2)$ time.
    \end{restatable}

We also show how to use these results, \autoref{thm:recmain} and \autoref{thm:recovery2}, black box to derive further recovery and solving results. For instance, using our earlier results we give nearly linear time solvers for inverse M-matrices and Laplacian pseudo-inverse systems which are summarized below.

\begin{restatable}[\textbf{$M$-Matrix Recovery and Inverse $M$-Matrix Solver}]{thm}{thmmmatrix}
\label{thm:mmatrix}
Let $\ma$ be the inverse of some unknown invertible symmetric $M$-matrix and let $\kappa$ be the ratio of the largest and smallest entries of $\ma \vones$. For any $\epsilon \in (0, 1/20)$ \autoref{alg:algommat} recovers an $(1-O(\epsilon))$-spectral approximator for $\ma^{-1}$ in 
$
\otilde (\epsilon^{-O(1)} n^2 ) 
$ time and consequently we can solve linear systems in $\ma$ in 
$\otilde(n^2)$ time.
\end{restatable}

\begin{restatable}[\textbf{Laplacian Recovery and Laplacian Pseudoinverse Solver}]{thm}{thminvlap}
\label{thm:laplacian}
Let $\ma$ be the pseudo-inverse of some unknown graph Laplacian $\glap$ and assume edge weights of graph $G$ are polynomially-bounded in $n$. For any $\epsilon \in (0, 1/20)$ \autoref{alg:algolap} recovers an $(1-O(\epsilon))$-spectral approximator for $\ma^{\dagger}$ in time 
$
\tilde{O} ( \epsilon^{-O(1)} n^2 ).  
$ and consequently we can solve linear systems in $\ma$ in $\otilde(n^2)$ time.
\end{restatable}

A key tool we leverage to prove these results are interesting structural facts about the inverse of SDD matrices and M-matrices proved in \autoref{sec:appproofs}.

\section{Solvers and Structured Recovery Algorithms}
\label{sec:solvers}

In this section we prove our main results, leveraging our algorithm for the spectral approximation problem (\autoref{prob:recovery}). In particular we show how to use \autoref{thm:recmain} to recover spectral approximations of non-negative combinations of PSD matrices with a nearly-constant number of adaptive linear measurements, and further we use these to obtain faster algorithms for solving several new classes of matrices.

We divide our results as follows: In \autoref{subsec:solver1} we prove \autoref{thm:recovery2} and provide linear system solver and solve recovery problem for perturbed Laplacians. In \autoref{subsec:solver2} we demonstrate that in the special case where the matrix is a symmetric $M$-matrix or a Laplacian we can solve the sparse recovery problem given only access to the matrice's \emph{inverse}. As a corollary, we also obtain faster algorithms for solving linear systems in the inverses of Laplacians and M-matrices. 

\newcommand{\logcond}{ }
\subsection{Matrix Compressed Sensing}\label{subsec:solver1}
The main result of this subsection is the following:
\thmmatrecmain*

We prove this result by leveraging a simple but powerful technique used throughout  previous literature \cite{LiMP13,KapralovLMMS14,ajss18}  on solving linear system by solving a sequence of regularized linear systems. This technique is leveraged in our algorithm for proving the result, \autoref{algo2}. Let $\mb_{i} \defeq \frac{2^i}{\lambda} \mb + \md$, where $\md \defeq \sum_{j=1}^{d} \beta_j \mm_j$. In each iteration $i$, this algorithm makes a call to \autoref{thm:recmain} to recovery a matrix $\mH_i$ that is a linear combination of $\{ \mm_{j}\}_{j=1}^{d}$ and satisfies $1/20 \edit \mb_{i}\preceq \mH_{i} \preceq \mb_i$. Further the matrix $\mH_i$ also satisfies $1/40 \edit \mb_{i+1}\preceq \mH_{i} \preceq \mb_{i+1}$ and can be used as a preconditioner to implement a linear system solver for $\mb_{i+1}$ needed for the next iteration. Using \autoref{lem:precond} at each iteration $i+1$, we implement $\runtime_{SQ}$ using $\otilde(\frac{1}{\edit})$ calls to both $\runtime_{root}$ and linear system solves in $\mb_{i+1}$. This algorithm runs $u$ number of iterations, where $u$ is such that $2^{u}=\lambda/\mu$ and $u \in  \Theta(\log(\lambda / \mu))$ and recovers a matrix $\mH_u$ that satisfies $1/20 \edit \mb_{u}\preceq \mH_{u} \preceq \mb_u$. Since the number of iterations $u$ is $\Theta(\log(\lambda / \mu))$ and $2^u=\lambda/\mu$, $\mu \mb_{u}$ satisfies $ \mb \preceq \mu \mb_{u} \preceq 2 \mb$. Combining both $\mu \mH_{u}$ satisfies $1/20 \edit \mb\preceq \mu \mH_{u} \preceq 2 \mb$ and ${\mu} \mH_{u}$ acts as a preconditioner for solving linear system in $\mb$. In the final step, using a linear system solver for $\mb$ we again invoke \autoref{thm:recmain} to recover the final desired matrix $\mH$ that is a linear combination of $\{ \mm_{j}\}_{j=1}^{d}$ and satisfies $(1-O(\epsilon)) \edit \mb\preceq \mH \preceq \mb$. Below we make this argument formal and prove \autoref{thm:recovery2}.

\begin{algorithm}[h!] \begin{algorithmic}[1]
\caption{$\mathtt{SpectralApproximate}$($\mathcal{M},\mb,\beta,\lambda,\mu,\epsilon$)}
    \State $u \in \Theta(\log(\lambda / \mu))$;
    \State $\md = \sum_{j=1}^{d} \beta_j \mm_j$;
    \State $\mb_i = \frac{2^i}{\lambda} \mb + \md$;
    \State $\mH_0 = \md$;

    \For{$i = 1, 2, ..., u$}
    \State $\mb_i = \frac{2^i}{\lambda} \mb + \md$;
    \State $f_i(x) \gets \mb_i x$ access oracle for $\mb_i$;
    \State $g_i(x) \gets$ linear system solver for $\mb_i$ by using $\mH_{i-1}$ as a preconditioner;
     \State $h_i(x) \gets$ routine computing $\mc x$ for matrix $\mc$ such that $\mc \mc^{\top}=\mb_{i}$;
    \State $\mH_i \gets$ matrix that satisfies $1/20 \edit \mb_i \preceq \mH_{i} \preceq  \mb_i$ via \autoref{thm:recmain} with $f_i, g_i, h_i$;
    \EndFor
    \State $f(x) \gets \mb x$ access oracle for $\mb$;
    \State $g(x) \gets$ linear system solver for $\mb$ using ${\mu} \mH_u$ as a preconditioner;
     \State $h_i(x) \gets$ routine  computing $\mc x$ for matrix $\mc$ such that $\mc \mc^{\top}=\mb$;
    \State $\mH \gets$ matrix that satisfies $(1-O(\epsilon)) \edit \mb\preceq  \mH \preceq  \mb$ via \autoref{thm:recmain} with $f,g,h$;
    \State \textbf{Return} $\mH$.
\label{algo2}
\end{algorithmic}
\end{algorithm}

\begin{proof}[Proof of \autoref{thm:recovery2}]
We first prove by induction that for all $i \geq 0$ it is the case that iteration $i$ of \autoref{algo2} can be implemented to run in $\tilde{O}(\edit^{-O(1)} (\runtime_{\mb} + \runtime_{solve}+\runtime_{root} + \runtime_{MV} + \runtime_{QF}))$ time and that this iteration results in a matrix $\mH_{i}$ that is a non-negative linear combinations of $\{\mm_{j}\}_{j=1}^{d}$ and satisfies $1/20 \edit \mb_i\preceq \mH_{i} \preceq  \mb_i$.

For the base case, note that $\mH_0=\md$ is a non-negative linear combination of $\{\mm_{j}\}_{j=1}^{d}$ and satisfies $ \frac{1}{2}\edit \mb_0 \preceq \mH_0 \preceq \mb_{0}$ as desired.

Now suppose $\mH_{i}$ satisfies $\frac{1}{20}\edit \mb_{i} \preceq \mH_{i} \preceq \mb_{i}$, for some $i \geq 0$. We show this suffices to prove the inductive hypothesis for $i + 1$. First note that by design, $\frac{1}{2}\mb_{i+1} \preceq \mb_i \preceq \mb_{i+1}$ and therefore $\frac{1}{40}\edit \mb_{i+1} \preceq \mH_{i} \preceq \mb_{i+1}$. Using preconditioned Richardson \autoref{thm:precon_rich} we can implement $g_{i+1}$ using $\otilde(1/\edit)$\footnote{Recall from preliminaries, we use $\otilde$ notation to hide terms logarithm in condition number of PSD matrices.} calls to both $f_{i+1}$ and linear system solver in $\mH_{i}$. Further implementing $f_{i+1}$ only requires one matrix vector product access to $\mb$ and $\md$ that takes $\tilde{O}(\runtime_{\mb} + \runtime_{MV})$ time. Using \autoref{lem:precond}, we can implement $h_{i+1}$, meaning compute $\mc x$ for matrix $\mc$ such that $\mc \mc^{\top}=\mb_{i+1}$ using $\otilde(1/\edit)$ calls to linear system solver in $\mb_{i+1}$ and $\runtime_{root}$. Combining all this analysis with \autoref{thm:recmain}, we can recover a matrix $\mH_{i+1}$ that is a linear combination of $\{ \mm_{j}\}_{j=1}^{d}$ and satisfies $\frac{1}{20}\edit \mb_{i+1} \preceq \mH_{i+1} \preceq \mb_{i+1}$ in time $\tilde{O}(\edit^{-O(1)}(\runtime_{\mb} + \runtime_{solve}+\runtime_{root} + \runtime_{MV} + \runtime_{QF}))$. 

Consequently, by induction the inductive hypothesis holds for all $i \geq 0$. This nearly completes the proof as in the final step of \autoref{algo2} we invoke \autoref{thm:recmain} with $\epsilon$ accuracy and recover a matrix $\mH$ that is a linear combination of $\{ \mm_{j}\}_{j=1}^{d}$ and satisfies $(1-O(\epsilon))\edit \mb \preceq \mH \preceq \mb$. All that remains is to formally bound the running time which in turn depends on the quality of $\mH_u$ as a preconditioner for $\mb$.

To bound the running time, note that in each iteration we spend at most  $\tilde{O}(\edit^{-O(1)} (\runtime_{\mb} + \runtime_{MV} + \runtime_{QF} + \runtime_{solve}+\runtime_{root}))$ time and there are only $u \in \Theta(\log(\frac{\lambda}{\mu}))$ iterations. Therefore in $\tilde{O}(\edit^{-O(1)} (\runtime_{\mb} + \runtime_{MV} + \runtime_{QF} + \runtime_{solve}))$ time (hiding $\log(\lambda/\mu)$ iteration term) \autoref{algo2} constructs matrix $\mH_u$, that satisfies $\frac{1}{20}\edit \mb_{u} \preceq \mH_{u} \preceq \mb_{u}$ and further $\mu \mH_u$ satisfies $\frac{1}{20}\edit \mb \preceq \mu \mH_{u} \preceq 2 \mb$. Similar to other iterations, we can implement the final iteration in time $\tilde{O}((\edit \epsilon)^{-O(1)} (\runtime_{\mb} + \runtime_{MV} + \runtime_{QF} + \runtime_{solve}+\runtime_{root}))$. Therefore, we can recover the desired matrix $\mH$ that is a linear combination of $\{ \mm_{j}\}_{j=1}^{d}$ and satisfies $(1-O(\epsilon))\edit \mb \preceq \mH \preceq \mb$ in $\tilde{O}((\edit \epsilon)^{-O(1)} (\runtime_{\mb} + \runtime_{MV} + \runtime_{QF} + \runtime_{solve}+\runtime_{root}))$ time.
\end{proof}

Using \autoref{thm:recovery2} we obtain our main result on solving perturbed Laplacians proved below. 
\corperturbedsolver*
\begin{proof}

Let $\mc = \ma + \lambda_{\min} \eye$ and observe that $\edit  \mc\preceq (\lap + \lambda_{\min} \eye) \preceq \mc$. Further, note that $\lap + \lambda_{\min} \eye$ is a positive linear combination of edge Laplacians $\lap_{ij} = (\ve_i - \ve_j) (\ve_i - \ve_j)^\top$ and diagonal matrices $\ve_i \ve_i^\top$. Further $\frac{1}{\lambda_{\min}} \gamma \mc \preceq \sum_{i}\ve_i \ve_i^\top \preceq \frac{1}{\lambda_{\max}} \mc$ and by applying \autoref{thm:recovery2} with $\epsilon = \frac{1}{2}$, we can recover a matrix $\md$ which is a nonnegative linear combination of $\lap_{ij}$ and $\ve_i \ve_i^\top$ and satisfies $\frac{1}{2}\edit \mc \preceq \md \preceq \mc$ in time $\tilde{O}( \logcond \edit^{-O(1)}(\runtime_\mc + \runtime_{QF} + \runtime_{MV} +\runtime_{solve}+\runtime_{root}))$.

We claim that all five of these oracles can be implemented in $\tilde{O}(n^2)$ time. Trivially, computing matrix vector products with $\mc$ takes $O(n^2)$ time (i.e. $\runtime_{\mc} = O(n^2)$ and multiplying a linear combination of $\lap_{ij}$ {and $\ve_i \ve_i^\top$} by a vector can also be implemented in $O(n^2)$ time (i.e. $\runtime_{MV} = O(n^2)$). It is easy to see that for any $\vx$ computing $\vx^\top \lap_{ij} \vx = (\vx(i) - \vx(j))^2$ for all $i,j$ and and $\vx^\top \ve_i \ve_i^\top \vx = \vx(i)^2$ for all $i$ can be carried out in $O(n^2)$ time ($\runtime_{QF} = O(n^2))$. 

As nonnegative linear combinations of $\lap_{ij}$ {and $\ve_i \ve_i^\top$} is clearly symmetric diagonally-dominant (SDD), we can solve linear systems in SDD matrices in $\otilde(n^2)$ by \cite{SpielmanT04} ($\runtime_{solve} = \otilde(n^2)$)). Further for any non-negative vector $\alpha$, such that $\mb=\sum_{ij} \alpha_{ij} \lap_{ij}+\sum_{i} \alpha_i \ve_i \ve_i^\top$ is a PD matrix also implies $\mb$ is a SDDM matrix and using Theorem 3.4 in \cite{CCLPT14} we can implement $\runtime_{root}$ for SDDM matrices in $\otilde(n^2)$ time.

Substituting these runtimes gives our desired claim that $\md$ can be found in $\otilde{(\logcond\edit^{-O(1)} n^2)}$ time. Further, by \autoref{thm:precon_rich} using $\md$ as a preconditioner we can solve linear systems in $\mc$ in only $\otilde(\logcond\edit^{-O(1)}n^2)$ time.

We now consider using the linear system solver in $\mc$ to solve the equation $\ma \vx = \vb$ provided there exists a solution. Just as in \autoref{thm:precon_rich}, we consider the performance of the preconditioned Richardson iteration. Let us initialize $\vx_0 = \vzeros$ and run the iteration $\vx_{k+1} = \vx_k - \mc^{-1} (\ma \vx_k - \vb)$ for all $k \geq 0$. We observe that for any $k$ we have $\ma \vx_{k+1} - \vb = (\mi - \ma \mc^{-1} )(\ma \vx_k - \vb)$. Now, we observe
\begin{align*}
\ma \vx_{k} - \vb = (\mi - \ma \mc^{-1}) (\ma \vx_{k-1} - \vb) = (\mi - \ma \mc^{-1})^k (\ma \vx_{0} - \vb) = (\mi - \ma \mc^{-1})^k (-\vb).
\end{align*}
by inducting over $k$. As $\vb$ lies inside the range of $\ma$, we note that for any vector $\vv$ where $\mc \vv = \lambda_{\min} \vv$ we have $\ma \vv = \vzeros$ and therefore $\vb^\top \vv = 0$. Thus if $\vv$ is an eigenvector of $\mc$ that is not orthogonal to $\vb$, then it must correspond to an eigenvalue of at least $2 \lambda_{\min}$ (because minimum non-zero eigenvalue of $\ma$ is at least $\lambda_{\min}$). The previous statement implies, if $\vv$ is an eigenvector of $\mi - \ma \mc^{-1}$ that is not orthogonal to $\vb$, then it must correspond to an eigenvalue of at most $1/2$. With this insight, it is not hard to see that 
\[
\|\ma \vx_{k} - \vb \|_2 = \|(\mi - \ma \mc^{-1})^k \vb\|_2  \leq \frac{1}{2^k} \|\vb\|_2.
\]
Thus, preconditioned Richardson applied to $\ma \vx = \vb$ yields a linearly convergent algorithm, i.e. by \autoref{thm:precon_rich} this means that we can solve linear systems in $\ma$ in $\otilde(1)$ linear solves against $\mc$ (which we argued cost $\otilde(\logcond\edit^{-O(1)}n^2)$ earlier) and $\otilde(1)$ matrix-vector products with $\ma$. Therefore, we can solve the linear system in $\otilde(\logcond\edit^{-O(1)}n^2)$ time and the first claim follows.

To show our second claim in this lemma, let $\lap$ be the Laplacian of some connected graph. In this case, define $\mc = \ma + \frac{1}{n}\vones\vones^{\top}$ and observe that $\edit  \mc\preceq (\lap  + \frac{1}{n}\vones\vones^{\top}) \preceq \mc$. Note that $\lap + \frac{1}{n}\vones\vones^{\top}$ is a non-negative linear combination of edge Laplacians $\lap_{ij}$ and $\vones\vones^{\top}$ and further $\frac{1}{\lambda_{\min}} \gamma \mc \preceq \frac{1}{n}\lap_{K_n}+\frac{1}{n}\vones\vones^{\top} \preceq \frac{1}{\lambda_{\max}} \mc$. Applying \autoref{thm:recovery2} with $\epsilon$, we can recover a matrix $\md$ which is a nonnegative linear combination of $\lap_{ij}$ and $\vones\vones^{\top}$ and satisfies $(1-O(\epsilon))\edit \mc \preceq \md \preceq \mc$ in time $\tilde{O}\left( \logcond (\edit \epsilon)^{-O(1)}(\runtime_\mc + \runtime_{QF} + \runtime_{MV} +\runtime_{solve}+\runtime_{root})\right)$. Note that $\lap'=\md-\alpha \vones\vones^{\top}$, for $\alpha$ such that $\vones\vones^{\top}$ belongs to kernel of $\lap'$ satisfies $(1-O(\epsilon))\edit \ma \preceq \lap' \preceq \ma$ and it is easy to see that such an $\alpha$ can be computed in $O(\runtime_{MV})$ time. Using similar analysis as the first claim, we can implement oracles $\runtime_\mc$, $\runtime_{QF}$ and $ \runtime_{MV}$ in $\tilde{O}(n^2)$ time. 

Further, implementing $\runtime_{solve}$ only requires $\otilde(n^2)$ time because for any matrix $\mb= \sum_{ij}\alpha_{ij} \lap_{ij}+\beta \vones\vones^{\top}$ ($\alpha_{ij},\beta \in\R_{\geq 0}$), $\mb$ is a sum of $\lap$ (Laplacian) and $\beta \vones\vones^{\top}$ that are orthogonal to each other. We can solve linear system $\mb x=b$, by solving $\lap x=b'$ and $\sqrt{\beta}\vones\vones^{\top} x =b''$, where $b''$ and $b'$ are orthogonal projections of $b$ onto $\vones$ and orthogonal subspace to $\vones$ respectively. Solving $\lap x=b'$ only takes $\otilde(n^2)$ time by $\cite{SpielmanT04}$. Also implementing $\runtime_{root}$, meaning computing $\mH x$ for matrix $\mH$ such that $\mH \mH^{\top}=\mc$ is equivalent to computing $(\mH_1+\frac{1}{\sqrt{n}} \vones \vones^{\top}) x$ for matrix $\mH_1$ such that $\mH_1 \mH_1^{\top}=\ma$ and carrying out $\mH_1x$ is a central result in \cite{CCLPT14} because $\ma$ is a SDD matrix, see Theorem 3.4 which implements $\runtime_{root}$ for SDDM matrices in $\otilde(n^2)$ and Lemma A.1 shows implementing $\runtime_{root}$ for SDD matrix, can be reduced to SDDM matrix that is twice as large and together give this result.
\end{proof}

\subsection{Inverse Symmetric $M$-Matrix Solver}\label{subsec:solver2}

In the previous subsection, we gave an algorithm to solve linear systems and recover spectral approximation to a PSD matrix $\ma$ using only $\tilde{O}(\epsilon^{-O(1)})$ linear measurements for the case of perturbed Laplacians. In this section, we achieve analogous results to different types of structured matrices. More specifically, we show that if $\ma$ is an invertible symmetric $M$-matrix or a Laplacian, we can recover a spectral approximation to it in $\tilde{O}(\epsilon^{-O(1)})$ linear measurements of its pseudo-inverse. We next formally state these two main results. These results also provide linear system solvers for Inverse $M$-matrices and Laplacian pseudo-inverses.
\thmmmatrix*
\thminvlap*

To prove our first theorems of this section, we need the following two useful lemmas about $M$-matrices. Their proofs are deferred to \autoref{sec:appproofs}.

\begin{restatable}{lemma}{lemmamfactone}
\label{lemma:mfact1}
Let $\mm$ be an invertible symmetric $M$-matrix. Let $\vx = \mm^{-1} \vones$ and define $\mx$ to be the diagonal matrix with $\vx$ on its diagonal. Then $\mx \mm \mx$ is a SDD matrix with nonpositive off-diagonal. 
\end{restatable}

\begin{restatable}{lemma}{lemmamfacttwo}
\label{lemma:mfact2}
Let $\ma$ be an invertible SDD matrix with nonpositive off-diagonal. For any $\alpha \geq 0$, the matrix $\mb = (\ma^{-1} + \alpha \eye)^{-1}$ is also a SDD matrix with nonpositive off-diagonal. 
\end{restatable}

With these facts, we sketch out the proof of the first of our two claimed theorems. We design our algorithm with a similar blueprint to the one we described in the previous section. Since we have access to $\ma = \mm^{-1}$ for some invertible symmetric $M$-matrix $\mm$, we can trivially compute $\vx = \mm^{-1} \vones = \ma \vones$. By \autoref{lemma:mfact1}, $\mx \mm \mx$ is a symmetric diagonally-dominant matrix with nonpositive off-diagonal, we observe that $\ma = \mx \mx^{-1} \mm^{-1} \mx^{-1} \mx = \mx (\mx \mm \mx)^{-1} \mx$ and  $\mx^{-1} \ma \mx^{-1}=(\mx \mm \mx)^{-1}$ is the inverse of a symmetric diagonally-dominant matrix with nonpositive off-diagonal. Now if we were able to compute a spectral approximation $\mb$ to $\mx^{-1} \ma \mx^{-1}$, then clearly $\mx \mb \mx$ is a spectral approximation to $\ma$. 

Define $u  \in \Theta( \log(\lambda_{\max} \kappa/\lambda_{\min}))$, such that $2^{u}=\frac{\lambda_{\max} \kappa^2}{\lambda_{\min}}$. Our algorithm will maintain a series of matrices $\mb_i = \frac{2^i}{\lambda_{\max} \kappa} \mx^{-1} \ma \mx^{-1} + \eye$. We observe by \autoref{lemma:mfact2} that for each $i$, $\mb_i$ is the inverse of a symmetric diagonally dominant matrix with nonpositive off-diagonal. Thus, we have that for any $i$, $\mb_i^{-1}$ is a nonnegative linear combination of edge Laplacians $\lap_{ij}$ and diagonal matrices $\ve_i \ve_i^\top$. Let $\mH_{0}=\mi$ and note that $\mH_0$ is a $1/20$-spectral approximation to $\mb_{0}^{-1}$ and further combined with $1/2\mb_{1} \preceq\mb_{0} \preceq \mb_{1}$ (by design) implies $\mH_{0}$ is a $1/40$-spectral approximation to $\mb_{1}^{-1}$. Using \autoref{thm:precon_rich}, we can solve linear systems in $\mb_1$ using $\otilde(1)$ matrix vector product calls to $\mb_1$ and $\mH_{0}$. Invoking \autoref{thm:recovery2} we can recover matrix $\mH_{1}$ that is a $1/20$-spectral approximation to $\mb_1^{-1}$.
Repeating this argument eventually gives us $\mH_u$, a $1/20$-spectral approximation to $\mb_u^{-1}$. As $\frac{\lambda_{\min}}{\kappa} \mb_u$ is spectrally within a factor of $2$ of $\mx \ma^{-1} \mx$, using \autoref{thm:precon_rich} we can solve linear systems in $\mx^{-1} \ma \mx^{-1}$ with only $\otilde(1)$ matrix vector product calls to $\mx^{-1} \ma \mx^{-1}$ and $\mH_{u}$: another application of \autoref{thm:recovery2} yields our desired spectral approximation to $(\mx^{-1} \ma \mx^{-1})^{-1}$ and consequently a spectral approximation to $\ma^{-1}$. Formally, we now prove \autoref{thm:mmatrix}:

\begin{algorithm} \begin{algorithmic}[1]
    \caption{$\mathtt{MMatrixInvSolver}$($\ma,y$)}
        \State $u \in \Theta(\log(\lambda_{\max} \kappa/ \lambda_{\min}))$;
        \State $\vx = \ma \vones$;
        \State $\mb = \mx^{-1} \ma \mx^{-1}$;
        \State $\mH_{0} = \eye$;
        \For{$i = 1, 2, ..., u$}
        \State $\mb_i = \frac{2^i}{\lambda_{\max} \kappa} \mb + \eye$;
        \State $g_i(x) \gets$ linear system solver for $\mb_i$ using $\mH_{i-1}^{-1}$ as a preconditioner;
        \State $\mH_i \gets \frac{1}{20}$-spectral approximator for $\mb_i^{-1}$ via \autoref{thm:recovery2} on basis set $\lap_{ij}$, $\ve_i \ve_i^\top$ using $g_i$;
        \EndFor
        \State $g(x) \gets$ linear system solver for $\mb$ using $\lambda_{\min} \mH_u^{-1}$ as a preconditioner;
        \State $\mH \gets (1-\epsilon)$-spectral approximator for $\mb^{-1}$ constructed via \autoref{thm:recovery2} on basis set $\lap_{ij}$, $\ve_i \ve_i^\top$ using $g$;
        \State \textbf{Return} $\mx^{-1} \mH \mx^{-1}$.
    \label{alg:algommat}
    \end{algorithmic}
    \end{algorithm}

\begin{proof}[Proof of \autoref{thm:mmatrix}]
We first realize the work required to compute $\mx$ is clearly $\otilde(n^2)$, as it can be found in a single matrix-vector product.

Just as in our algorithm for \autoref{thm:recovery2}, at every step $i$ we inductively maintain $\mH_i$, a $1/20$-spectral approximation of $\mb_i^{-1}$ and hence a $1/40$-spectral approximation to $\mb_{i+1}^{-1}$. By \autoref{thm:precon_rich} at every step $i$, we can implement $g_i$ using $\otilde(1)$ matrix vector product calls to $\mb_{i}$ and $\mH_{i-1}$ which can be implemented in time $\otilde(\runtime_{\mb}+\runtime_{MV}) \in \otilde(n^2)$. Since we can compute quadratic forms in $\lap_{ij}$ for all $i,j$ and $\ve_i \ve_i^\top$ for all $i$ in $\otilde(n^2)$ time, we can implement the $\runtime_{QF}$ oracle at every step in $\otilde(n^2)$ time. Further, matrix vector product with any nonnegative linear combinations of $\lap_{ij}$ and $\ve_i \ve_i^\top$ ($\runtime_{MV}$) can be trivially computed in $\otilde(n^2)$ time. As before, nonnegative linear combinations of $\lap_{ij}$ and $\ve_i \ve_i^\top$ is a SDD matrix, therefore $\runtime_{solve}$ can be performed in $\otilde(n^2)$ time by \cite{SpielmanT04}. Further for any non-negative vector $\alpha$, such that $\mb=\sum_{ij} \alpha_{ij} \lap_{ij}+\sum_{i} \alpha_i \ve_i \ve_i^\top$ is a PD matrix also implies $\mb$ is a SDDM matrix and using Theorem 3.4 in \cite{CCLPT14} we can implement $\runtime_{root}$ in $\otilde(n^2)$ time. Combining all the analysis we can implement each iteration of our algorithm in $\otilde(n^2)$ time.

As our algorithm performs the above for $u \in \Theta(\log(\lambda_{\max} \kappa /\lambda_{\min}))$ different $i$, we can compute $\mH_u$ in $\otilde(u n^2)$ time. For similar reasons to the above, we can then recover $\mH$ in time
$
\tilde{O} ( \epsilon^{-O(1)} n^2 )
$ (hiding the iteration term). Computing the returned $\mx^{-1} \mH \mx^{-1}$ takes only $\otilde(n^2)$ more time-- this is our claimed spectral approximator, and preconditioning linear system solves in $\ma$ with this for $\epsilon = 1/20$ gives our claimed running time for solving in $\ma$. 
\end{proof}

Similar to the result above, in the remainder of this section we show how to solve linear systems in Laplacian pseudo-inverses. First we briefly describe our approach. Let $\ma$ be our matrix, and let us make the simplifying assumption that $\ma$ is the laplacian pseudo-inverse of a connected graph. Let us define $\ma_i = \frac{2^i}{\lambda_{\max}} \ma + 1/n \lap_{K_n}$, where $\lap_{K_n}$ is the Laplacian of the complete graph on $n$ nodes. We observe the following useful fact about $\ma_i$:

\begin{lemma}
    For any $i$, $\ma_i^\pseudo$ is a Laplacian matrix.
\end{lemma}

\begin{proof}
$\ma_i^\pseudo$ is obviously symmetric and positive semidefinite. Further, it is also clear that it has the all-ones vector in its kernel. By standard algebraic manipulation and the Woodbury Matrix formula we have,
\begin{align*}
\ma_i^\pseudo &= \left(\frac{2^i}{\lambda_{\max}} \ma + \frac{1}{n}\lap_{K_n}\right)^\pseudo = \left(\frac{2^i}{\lambda_{\max}} \ma + \frac{1}{n}\lap_{K_n} + \frac{2}{n} \vones \vones^\top\right)^{-1} - \frac{1}{2n} \vones \vones^\top \\
  &= \left(\frac{2^i}{\lambda_{\max}} \ma + \frac{1}{n} \vones \vones^\top + \eye\right)^{-1} - \frac{1}{2n} \vones \vones^\top = \left(\left(\frac{\lambda_{\max}}{2^i} \ma^\pseudo + \frac{1}{n} \vones \vones^\top\right)^{-1} + \eye\right)^{-1} - \frac{1}{2n} \vones \vones^\top \\
              &= \eye -  \left(\frac{\lambda_{\max}}{2^i} \ma^\pseudo + \frac{1}{n} \vones \vones^\top + \eye \right)^{-1} - \frac{1}{2n} \vones \vones^\top = \eye -  \left(\frac{\lambda_{\max}}{2^i} \ma^\pseudo + \eye \right)^{-1}.
\end{align*}
We observe that $\frac{\lambda_{\max}}{2^i} \ma^\pseudo + \eye$ is a positive definite symmetric diagonally-dominant matrix: its inverse is entry wise nonnegative and thus $\ma_i^\pseudo$ has nonpositive off-diagonal. Combining these facts we can conclude that $\ma_i^\pseudo$ is a Laplacian.
\end{proof}

Let $u \in \Theta(\log (\lambda_{\max}/\lambda_{\min}) )$ be such that $2^{u}=\lambda_{\min}/\lambda_{\max}$. With the insight from previous lemma, we define $\mb_i = \ma_i + 1/n \vones \vones^\top$: observe that $\mb_i$ is positive definite and $\mb_i^{-1} = \ma_i^\pseudo + 1/n \vones \vones^\top$. Thus for any $i$, $\mb_i^{-1}$ is a nonnegative linear combination of edge Laplacians $\lap_{ij}$ and $\vones \vones^\top$ and matrix vector product with $\mb_{i}$ can be computed in $\otilde(n^2)$ time. Further, $\mH_0^{-1}=\mi$ is spectrally within a factor of $2$ of $\mb_{0}$.
Just as we did in the earlier section, we note that $\mH_0$ is also a $1/40$-spectral approximation to $\mb_1^{-1}$ and we can solve linear systems in $\mb_{1}$ with $\otilde(1)$ calls to matrix vector product to $\mb_1$ and $\mH_0$. Repeating this argument we can eventually recover $\mH_u$, a $1/20$-spectral approximation to $\mb_u^{-1}$. Let $\alpha$ be such that $\mm_u = \mH_u - \alpha \vones \vones^\top$ has the all-ones vector in its kernel. Clearly $\mm_u$ is a $1/20$-spectral approximation to $\ma_u$. But now, as $\lambda_{\min} \ma_u$ is spectrally within a factor of $2$ of $\ma$ itself,  we conclude that $\lambda_{\min} \mm_u$ can be used to solve linear systems in $\ma$. This idea is codified below and explicit running time for each iteration is provided in the proof below:

\begin{algorithm} \begin{algorithmic}[1]
    \caption{$\mathtt{LapInvSolver}$($\ma,y$)}
        \State $u \in \Theta(\log(\lambda_{\max} / \lambda_{\min}))$;
        \State $\mH_{0} = \eye$;
        \State $\mb=\ma+1/n\vones \vones^{\top}$;
        \For{$i = 1, 2, ..., u$}
        \State $\ma_i = \frac{2^i}{\lambda_{\max}} \ma + 1/n \lap_{K_n}$;
        \State $\mb_i = \ma_i + 1/n \vones \vones^\top$;
        \State $g_i(x) \gets$ linear system solver for $\mb_i$ by using $\mH_{i-1}^{-1}$ as a preconditioner;
        \State $\mH_i \gets \frac{1}{20}$-spectral approximator for $\mb_i^{-1}$ via \autoref{thm:recovery2} on basis set $\lap_{ij}$, $\vones \vones^\top$ using $g_i$;
        \EndFor
        \State $g(x) \gets$ linear system solver for $\mb$ using $\lambda_{\min} \mH_u^{-1}$ as a preconditioner;
        \State $\mH \gets (1-\epsilon)$-spectral approximator for $\mb^{-1}$ constructed via \autoref{thm:recovery2} on basis set $\lap_{ij}$, $\vones \vones^\top$ using $g$;
        \State $\alpha \gets$ value such that $\mH - \alpha \vones \vones^\top$ has the all-ones vector in its kernel;
        \State \textbf{Return} $\mH - \alpha \vones \vones^\top$.
    \label{alg:algolap}
    \end{algorithmic}
    \end{algorithm}
    
\begin{proof}[Proof of \autoref{thm:laplacian}]
This proof is essentially the same as the other proofs given in this section. Since $\mH_{i-1}$ is $1/40$-spectral approximation to $\mb_{i}^{-1}$. We can solve linear systems in $\mb_i$ using $\otilde(1)$ matrix vector product calls to $\mb_{i}$ and $\mH_{i-1}$ by \autoref{thm:precon_rich} and $g_i$ can be implemented in time $\otilde(\runtime_{\ma}+\runtime_{MV}) \in \otilde(n^2)$. Note that $\mH_{i-1}$ is a non-negative linear combination of $\lap_{ij}$ and $\vones \vones^{\top}$ which we already discussed in perturbed laplacian case. We can compute quadratic forms in all the $\lap_{ij}$ and $\vones \vones^\top$ in $\otilde(n^2)$ time, and we can also compute matrix vector product with nonnegative linear combinations of these matrices in $\otilde(n^2)$ time. Recall again from the proof of perturbed laplacian, we can also implement $\runtime_{solve}$ and $\runtime_{root}$ for a non-negative linear combination of $\lap_{ij}$ and $\vones \vones^\top$ in time $\otilde(n^2)$. Now combining all the analysis and applying \autoref{thm:recovery2}, we can recover matrix $\mH_{i}$ which is a $1/20$-spectral approximation to $\mb_{i}^{-1}$.
We can thus compute $\mH$ in $\tilde{O} ( \epsilon^{-O(1)} n^2 ).  
$ time (hiding the iteration term). $\mH$ is spectrally close to $\mb^{-1}$ and it is not hard to see that $\mH-\alpha \vones \vones^{\top}$ for $\alpha$, such that, $\mH - \alpha \vones \vones^\top$ has the all-ones vector in its kernel is spectrally close to $\ma^{\dagger}$. Further, computing $\alpha$ and $\mH - \alpha \vones \vones^\top$ can both be done in trivial $\otilde(n^2)$ time, and thus we can compute $(1-\epsilon)$-spectral approximation to Laplacian pseudo-inverse of connected graphs in that time. Additionally we can use this to solve linear systems in Laplacian pseudo-inverses by again choosing $\epsilon = 1/20$ and using the output of the algorithm as a preconditioner. 

To finally prove our claimed result, we must finally show how to deal with the case where $\ma^\pseudo$ is the Laplacian of a disconnected graph. Fortunately, this is trivial. We note that when $i$ and $j$ lie in different connected components of the graph underlying $\ma$, $\ma_{ij} = 0$. In addition, whenever $i$ and $j$ lie in the same connected component there exists a path of indices $k_1, k_2, ... , k_d$ where $\ma_{i,k_1}, \ma_{k_1, k_2}, ... \ma_{k_d, j}$ are all nonzero. With this fact, it is easy to see how we can compute the connectivity structure of the graph underlying $\ma$ in $\otilde(n^2)$ time. Once we have these connected components, proving our result is simply a matter of permuting $\ma$ such that each connected component's inverse Laplacian appears on the block diagonal and then running our algorithm on each submatrix separately-- the running time follows trivially. 
\end{proof}

\newcommand{\TwoSidedOracle}{\mathtt{TwoSidedApproximator}}

\section{Building Spectral Approximations}
\label{sec:recovery}

The main goal of this section is to provide an algorithm to solve \autoref{prob:recovery}, a fundamental spectral approximation problem introduced in \autoref{sec:results} and restated below for convenience:

\probrecovery*

The problem defined above is a special case of the more general problem of solving a mixed packing/covering semidefinite program (SDP). In a mixed packing/covering SDP problem we are guaranteed existence of $w \in \R^d_{\geq 0}$ (unknown) with $ \mb_1 \preceq \sum_{i=1}^{d}w_i \mm_{i} \preceq \mb_2$ and asked to find weights $w' \in \R^d_{\geq 0}$ such that $(1- O( \epsilon )) \mb_1 \preceq \sum_{i=1}^{d}w'_i \mm_{i} \preceq  \mb_2$. The mixed packing/covering SDP is a natural generalization of mixed packing/covering LP and nearly linear time algorithms are known for the LP case \cite{Young14}. In the SDP case, it is open how to produce such an algorithm but there are nearly linear time algorithms for pure packing SDP and covering SDP \cite{Allen-ZhuLO16,JY11,JY12,PT12} (but not for mixed packing/covering SDPs).

Nevertheless, in this section we provide an algorithm for (\autoref{prob:recovery}), which is a special case of packing covering SDPs when the left and right matrices are $\mb_1 = \gamma \mb$ and $\mb_2=\mb$ and the access to the input is restricted. The main result of this section is the following.

\thmrecmain*

Our proof of this theorem follows from careful adaptation, simplification, and extension of  \cite{ls17}. The work of \cite{ls17} was tailored towards producing linear sized spectral sparsifiers of an explicitly given matrix; however carefully viewed and modified its algorithm and analysis can be shown to solve \autoref{thm:recmain}. Since we are solving a slightly different problem then \cite{ls17}, we need to carefully modify all algorithms and analysis from \cite{ls17}. Consequently, in this section we state these modified version of algorithms and analysis from \cite{ls17}. 

The remainder of this section is organized as follows. The first three subsections, \autoref{subsec:prelim}, \autoref{subsec:oracle}, and \autoref{subsec:twosided} are dedicated to solve \autoref{prob:recovery} in the special case when matrix $\mb$ is identity. Our algorithm for this special case is called $\TwoSidedOracle$ and is same as the $\mathtt{Sparsify}$ algorithm in \cite{ls17} with few modifications. In  \autoref{subsec:prelim} we define a new potential function that is a generalization of one used in \cite{ls17} to handle our new \autoref{prob:recovery} and state important lemmas about it we use throughout. In \autoref{subsec:oracle} we provide and analyze properties of, $\oracle$, a critical subroutine invoked by $\TwoSidedOracle$ that is a modified version of the "one-sided oracle" in \cite{ls17}. In \autoref{subsec:twosided} we provide a full description of our algorithm $\TwoSidedOracle$ and prove it solves the identity case of \autoref{prob:recovery}. In \autoref{subsec:reduction} we then provide a standard reduction to solve \autoref{prob:recovery} for general $\mb$ using a solver for the identity case. This analysis, relies on an efficient implementation of $\oracle$ which we show how to achieve using Taylor approximations and Johnson-Lindenstrauss (JL) sketches in \autoref{sec:JL}.

\newcommand{\mlo}{\textbf{Y}}
\newcommand{\mlt}{\textbf{X}}
\newcommand{\mub}{\mi}
\newcommand{\mlb}{\gamma\mi}
\newcommand{\clb}{\gamma}
\renewcommand{\given}{\gamma\mi \preceq \sum_{i=1}^{d}w_{i}\mm_{i} \preceq \mi}
\newcommand{\eigenlb}{\ln^{-2} (2e^4n)}
\newcommand{\eigenlbsq}{\ln^{-1} (2e^4n)}
\renewcommand{\rt}{\tilde{O}\left(\frac{1}{\gamma^{4}\error^{O(1)}}\left(\runtime_{\mb} + \runtime_{\mb^{-1}} + \runtime_{MV} + \runtime_{QF}\right) \right)}
\subsection{Preliminaries}\label{subsec:prelim}

We begin by defining a generalization of exponential potential function introduced in \cite{ls17} and then provide several key properties of this potential function. The barrier is critical to our analysis and has several nice properties. Importantly, the fact that it blows up as the matrix approaches its upper or lower barriers is crucial to ensuring that intermediate matrices are well-conditioned. This barrier is also central to bounding the number of iterations needed for our algorithm to terminate. 

Our exponential potential function is defined as follows:
\begin{equation}\label{eq:defpotential}
\Phi_{u,\ell}(\ma)
\defeq
\Phi_{u}(\ma) + \Phi_{\ell}(\ma),
\end{equation}
where we call $\Phi_u$, and $\Phi_\ell$ the upper and lower barrier values respectively and
\[
\Phi_u(\ma)
\defeq \tr\left[\exp\left((u\mub - \ma)^{-1}\right)\right]
\text{ and }
\Phi_{\ell}(\ma)
\defeq \tr\left[\exp\left((\gamma^{-1}  \ma - \ell \mi )^{-1}\right)\right]~.
\]

The next lemmas are analogous to Lemma 3.1 and Lemma 4.3 in \cite{ls17}. We provide a general lemma, \autoref{lem:delta} and then use it to prove \autoref{lem:potentialchange} and \autoref{lem:potential2}, our central tools for understanding the potential function. These lemmas each measure change in the potential function value relative to the change in parameters $\ma, u$ and $\ell$. \autoref{lem:potentialchange} measures the change in potential due to addition of a matrix $\Delta$ to our current progress while keeping the barriers $u.\ell$ unchanged and \autoref{lem:potential2} measures the change in potential due to change in barriers while input matrix is unchanged. The following function $\phideriv : \R^{n \times n} \rightarrow \R^{n \times n}$ defined on PD matrix occurs frequently in our analysis
\[
\phideriv(\ma) \defeq \exp(\ma^{-1})  \ma^{-2} 
\]
and we use it throughout this section:

\begin{lem}
\label{lem:delta}
For $\delta \in [0,1/10]$, if $\mlt$ and $\mDelta \in \R^{n \times n}$ are symmetric PSD matrices satisfying $\mDelta\preceq\delta\mlt^{2}$ and $\mDelta \preceq \delta\mlt$, then
\[
\tr\left(\exp\left( (\mlt - \mDelta)^{-1} \right) \right) \leq  \tr\left(\exp\left( \mlt^{-1}\right) \right) +(1+2\delta) \tr\left(\phideriv(\mlt)\mDelta\right)
~.
\]
and
\[
\tr \left( \exp\left( (\mlt+\mDelta)^{-1} \right) \right)
\leq  \tr\left(\exp\left( \mlt^{-1}\right)\right) -(1-2\delta)\tr\left(\phideriv(\mlt)\mDelta\right)
~.
\]
\end{lem}

\begin{proof}
Since $\mzero \preceq \mDelta \preceq  \delta \mlt$ if we let $\mPi \defeq \mlt^{-1/2}\mDelta\mlt^{-1/2}$ then we have that $\mzero \preceq \mPi \preceq \delta \mi$. Further since $|\delta| < 1$ this implies that
\begin{align*}
(\mlt \pm \mDelta)^{-1}
&= \mlt^{-1/2}\left( \mi \pm \mPi \right)^{-1}\mlt^{-1/2}
= \mlt^{-1/2} \left( \mi \mp \mPi \left( \mi \pm \mPi \right)^{-1}   
\right) \mlt^{-1/2}
\\
&\preceq 
\mlt^{-1/2} \left( \mi \mp ( 1 \pm \delta)^{-1} \mPi   
\right) \mlt^{-1/2}
= \mx^{-1} \mp \frac{1}{1 \pm \delta} \cdot \mx^{-1} \mDelta \mx^{-1}
 ~.
\end{align*}
By the fact that $\tr(\exp(\cdot))$ is monotone under $\preceq$ and the Golden-Thompson Inequality, i.e. for symmetric $\ma,\mb$ it holds that $\tr(\exp(\ma + \mb)) \leq \tr(\exp(\ma) \exp(\mb))$, we have 
\begin{align*}
\tr \left(\exp\left( (\mlt \pm \mDelta)^{-1} \right)\right) 
& \leq \tr \left(\exp\left( \mlt^{-1} \mp \frac{1}{1 \pm \delta} \cdot \mlt^{-1}\mDelta\mlt^{-1} \right) \right) \\
& \leq \tr\left(\exp(\mlt^{-1}) \exp\left(\frac{\mp 1}{1 \pm \delta}\cdot \mlt^{-1}\mDelta\mlt^{-1}\right)\right).
\end{align*}
Since $0 \preceq \mDelta\preceq \delta \mlt^{2}$ and $\delta \leq 1/10$ by assumption, using the inequality that $e^x \leq 1 + x + x^2$ for all $x$ with $|x| < 1/2$  we have that
\[
\exp\left(\frac{\mp 1}{1 \pm \delta}\cdot \mlt^{-1}\mDelta\mlt^{-1}\right)
\preceq \mi + \left( \frac{\mp 1}{1 \pm \delta} + \frac{\delta}{(1 \pm \delta)^2} \right)\cdot \mlt^{-1}\mDelta\mlt^{-1}
\preceq \mi 
\mp 
\frac{1}{(1 \pm \delta)^2} \mlt^{-1}\mDelta\mlt^{-1}
~.
\]
Since, $\mp (1 \pm \delta)^{-2} \leq \mp (1 \mp 2\delta)$ for $\delta \in [0, 1/10]$ the result follows.
\end{proof}

\begin{lem}\label{lem:potentialchange}
Let $\ma$ be a symmetric matrix such that $\ell \mlb\prec\ma\prec u\mub$ for some $\gamma \in (0, 1)$ and $u,\ell \in \R$. If $\mDelta \in \R^{n \times n}$ is a symmetric PSD matrix with $\mDelta \preceq \delta \clb \alpha \mi$ for some $\delta \in [0, 1/10]$ and $\alpha \in (0,1]$ such that $\alpha \mi \preceq \lambda_{\min}(u \mi - \ma)^2$ and $\alpha \mi \preceq \lambda_{\min}(\clb^{-1}\ma-\ell \mi)^{2}$ then,
\[
\Phi_{u,\ell}(\ma+\mDelta)  \leq   \Phi_{u,\ell}(\ma) + (1+2\delta)\tr\left[
\phideriv(u\mub - \ma)
\mDelta\right]
-(1-2\delta)\tr\left[
\phideriv \left(\clb^{-1}\ma-\ell \mi \right)
\clb^{-1}
\mDelta
\right]
~.
\]
\end{lem}
\begin{proof}
First note that, for $\alpha \in (0,1]$, $\alpha \mi \preceq (u \mi - \ma)^2$ implies $\alpha \mi \preceq (u \mi - \ma)$ and using $\clb \leq1$ we get $\mDelta \preceq \delta (u \mi - \ma)$ and $\mDelta \preceq \delta (u \mi - \ma)^2$. Further $\alpha \mi \preceq (\clb^{-1}\ma-\ell \mi)^{2}$ implies $\alpha \mi \preceq (\clb^{-1}\ma-\ell \mi)$ and we also get $\frac{\mDelta}{\clb}\preceq\delta(\clb^{-1}\ma-\ell \mi)$.and $\frac{\mDelta}{\clb}\preceq\delta(\clb^{-1}\ma-\ell \mi)^2$.
Now recall,
$$\Phi_{u,\ell}(\ma+\mDelta)=\Phi_{u}(\ma+\mDelta)+\Phi_{\ell}(\ma+\mDelta)=\tr\left(\exp((u\mi-(\ma+\mDelta))^{-1})\right)+\tr\left(\exp(\clb(\ma+\mDelta-\ell \clb \mi)^{-1})\right)$$
We analyze terms $\exp((u\mi-(\ma+\mDelta))^{-1})$ and $\exp(\clb(\ma+\mDelta-\ell \clb \mi)^{-1})$ individually. For the first term, $\mlt=u\mi-\ma$ and $\Delta=\mDelta$ satisfies conditions of \autoref{lem:delta} because $\mDelta\preceq\delta(u\mub-\ma)^{2}$ and $\mDelta\preceq \delta(u\mub-\ma)$, yielding
\begin{equation}\label{equpper}
\begin{split}
\tr\left(\exp((u\mi-(\ma+\mDelta)^{-1})) \right)& \leq \tr\left(\exp\left( (u\mi-\ma)^{-1}\right)\right) + (1+2\delta) \tr\left(\phideriv(u\mi-\ma)\mDelta \right)\\
&= \Phi_{u}(\ma)+ (1+2\delta) \tr\left(\phideriv(u\mi-\ma)\mDelta \right) ~.
\end{split}
\end{equation}
Similarly for the second term, $\mlt=\clb^{-1}\ma-\ell \mi$ and $\Delta=\frac{\mDelta}{\clb}$ satisfies conditions of \autoref{lem:delta} because $\frac{\mDelta}{\clb}\preceq\delta(\clb^{-1}\ma-\ell \mi)^{2}$ and $\frac{\mDelta}{\clb}\preceq \delta(\clb^{-1}\ma-\ell \mi)$, yielding,
\begin{equation}\label{eqlower}
\begin{split}
\tr\left(\exp(\clb(\ma+\mDelta-\ell \clb \mi)^{-1}) \right)& \leq \tr\left(\exp\left( \clb(\ma-\ell \clb \mi)^{-1}\right)\right) - (1-2\delta) \tr\left(\phideriv\left(\clb^{-1}\ma-\ell \mi \right)\frac{\mDelta}{\clb} \right)\\
&= \Phi_{\ell}(\ma)- (1-2\delta) \tr\left(\phideriv\left(\clb^{-1}\ma-\ell \mi\right)\frac{\mDelta}{\clb} \right)
\end{split}
\end{equation}
Our lemma statement follows by combining \autoref{equpper} and \autoref{eqlower}.
\end{proof}

\begin{lem}\label{lem:potential2}
Let $\ma$ be a symmetric PSD matrix and let $\gamma \in (0, 1)$ and $u, \ell \in \R$ satisfy $\ell \mlb\prec\ma\prec u\mub$. If $\delta_u \in [0, \delta \alpha_{u}]$ and $\delta_{\ell} \in [0, \delta \alpha_{\ell}]$ for some $\delta\leq 1/10$, $\alpha_u \in (0,1]$ and $\alpha_{\ell} \in (0,1]$ such that $\alpha_{u} \leq  \lambda_{\min}(u\mub-\ma)^2$  and $\alpha_{\ell} \leq \lambda_{\min}(\clb^{-1}\ma-\ell \mi)^2$ for $\delta\leq 1/10$, then 
\[
\Phi_{u+\delta_u,\ell+\delta_{\ell}}(\ma)  \leq  \Phi_{u,\ell}(\ma)  - (1-2\delta) \delta_{u}\cdot \tr\left[\phideriv(u\mub-\ma)\right] 
+(1+2\delta)\delta_{\ell}\cdot \tr\left[
\phideriv\left(\clb^{-1}\ma-\ell \mi\right)
\right] ~.
\]
\end{lem}
\begin{proof}
First note that, for $\alpha_{u} \in (0,1]$, $\alpha_{\ell} \in (0,1]$, $\alpha_{u} \leq  \lambda_{\min}(u\mub-\ma)^2$ and $\alpha_{\ell} \leq  \lambda_{\min}(\clb^{-1}\ma-\ell \mi)^2$ imply $\alpha_{u} \leq  \lambda_{\min}(u\mub-\ma)$ and $\alpha_{\ell} \leq  \lambda_{\min}(\clb^{-1}\ma-\ell \mi)$ respectively. 
Now recall,
$$\Phi_{u+\delta_{u},\ell+\delta_{\ell}}(\ma)=\exp((u\mi +\delta_{u}\mi -\ma)^{-1})+\exp(\clb(\ma-\ell \clb \mi -\delta_{\ell}\clb\mi)^{-1})$$
We analyze terms $\exp((u\mi +\delta_{u}\mi -\ma)^{-1})$ and $\exp(\clb(\ma-\ell \clb \mi -\delta_{\ell}\clb\mi)^{-1})$ individually. For the first term, $\mlt=u\mi-\ma$ and $\Delta=\delta_{u}\mi$ satisfies conditions of \autoref{lem:delta} because $\delta_{u}\mi \preceq\delta(u\mub-\ma)^{2}$ and $\delta_{u}\mi\preceq \delta(u\mub-\ma)$, yielding
\begin{equation}\label{equpper2}
\begin{split}
\exp((u\mi+\delta_{u}\mi-\ma^{-1}))& \leq \tr\left(\exp\left( (u\mi-\ma)^{-1}\right)\right) - (1-2\delta) \tr\left(\phideriv(u\mi-\ma)\delta_{u}\mi \right)\\
&= \Phi_{u}(\ma)- (1-2\delta) \delta_{u}\tr\left(\phideriv(u\mi-\ma)\right) ~.
\end{split}
\end{equation}
Similarly for the second term, $\mlt=\clb^{-1}\ma-\ell \mi$ and $\Delta=\delta_{\ell}\mi$ satisfies conditions of $\autoref{lem:delta}$ because $\delta_{\ell}\mi\preceq\delta(\clb^{-1}\ma-\ell \mi)^{2}$ and $\delta_{\ell}\mi\preceq \delta(\clb^{-1}\ma-\ell \mi)$, yielding
\begin{equation}\label{eqlower2}
\begin{split}
\exp(\clb(\ma-\ell \clb \mi -\delta_{\ell}\clb\mi)^{-1})& \leq \tr\left(\exp\left( \clb(\ma-\ell \clb \mi)^{-1}\right)\right) + (1+2\delta) \tr\left(\phideriv\left(\clb^{-1}\ma-\ell \mi \right)\delta_{\ell}\mi \right)\\
&= \Phi_{\ell}(\ma)+ (1+2\delta)\delta_{\ell} \tr\left(\phideriv\left(\clb^{-1}\ma-\ell \mi\right) \right) ~.
\end{split}
\end{equation}
Our lemma statement follows by combining \autoref{equpper2} and \autoref{eqlower2}.
\end{proof}

\subsection{Modified One-sided Oracle}
\label{subsec:oracle}

In the subsection we define and show existence of our \emph{modified one-sided oracle} which is analogous to \emph{one-sided oracle} in \cite{ls17} and the key algorithmic primitive we use to solve \autoref{prob:recovery}. In \autoref{subsec:twosided} we show how to use this oracle through an iterative procedure to solve \autoref{prob:recovery}. In \autoref{subsec:reduction} we also discuss running time to implement a solver for modified one-sided oracle. We start this subsection by defining modified one-sided oracle.
\begin{defi}[\textbf{Modified one-sided oracle}]\label{deforacle}
Let  $\constb \in [0,1]$, $\cp \succeq \mzero$, and  $\cn \succeq \mzero$ be  symmetric matrices and let $\mathcal{M}=\{\mm_{i}\}_{i=1}^{d}$ be a set of matrices such that there exists weights $w \in \R^d_{\geq 0}$ with $\given$. 
We call an algorithm $\oracle\left(\calM, \constb,\cp, \cn \right)$ a modified one-sided oracle with speed $S \in (0,1]$ and error $\epsilon>0$, if  it  outputs  a matrix $\mDelta=\sum_{i=1}^{d}\alpha_{i}\mm_{i}$ such that 
\begin{enumerate}
\item $\Delta\preceq \constb \mi$ and $\alpha \in \R^{d}_{\geq 0}$.
\item $\cp\bullet \frac{\Delta}{\clb}-\cn\bullet \Delta \geq S\cdot \constb \left[ (1-\epsilon) \tr(\cp) - (1+\epsilon)\tr(\cn) \right]  ~.$
\end{enumerate}
\end{defi}
Our definition of \emph{modified one-sided oracle} differs from \emph{one-sided oracle} definition of \cite{ls17} in following ways: 1) We don't have any restrictions on  
$\nnz(\alpha)$. 2) Condition 1 is more restrictive and is asking $\Delta$ to be less than some scalar multiple of identity instead of any arbitrary matrix. 3) Condition 2 in our definition reweights $\cp$ to handle our more general case. 4) Condition 2 holds without expectation. 5) We don't have access to vector $w$ such that $\given$. 

Note for our purposes we only require Condition 2 to hold with high probability (as is true for \cite{ls17}). The modified one-sided oracle we produce only satisfies condition 2 with high probability and to simplify writing we just omit this high probability statement for Condition 2 everywhere.

We conclude this subsection by proving existence of a modified one-sided oracle with speed $S=1$ and error $\frac{\epsilon}{2}$.

\begin{lemma}[Existence of One-sided Oracle]\label{lem:sdp} For all $\constb \in [0, 1]$, there exists a modified one-sided oracle with speed $S=1$ and error $\epsilon = 0$.
\end{lemma}
\begin{proof}
Choose $\Delta=\constb\cdot \sum_{i=1}^{d}w_i\mm_{i}$. As $\clb\mi \preceq \sum_{i = 1}^{d} w_i \mm_i \preceq \mi$, clearly $\Delta \preceq \constb \mi$. Furthermore, 
\begin{align*}
\cp\bullet \frac{\Delta}{\clb}-\cn\bullet \Delta 
= \constb \left[ \cp\bullet \frac{\sum_{i=1}^{d}w_i\mm_{i}}{\clb}-\cn\bullet \sum_{i=1}^{d}w_i\mm_{i} \right] 
\geq \constb \left[\tr(\cp) -\tr(\cn) \right]~.
\end{align*}
\end{proof}

\subsection{Solving Identity Case with One-Sided Oracle}
\label{subsec:twosided}
In the previous subsection we defined modified one-sided oracle and showed existence of one with speed $S=1$ and error $0$. In this section, with the help of our exponential potential function we show how to use a one-sided oracle to implement our main algorithm $\TwoSidedOracle$.

The main algorithm $\TwoSidedOracle\left( \calM,\epsilon \right)$ works as follows:

\begin{algorithm}[H] \begin{algorithmic}[1]

\caption{$\TwoSidedOracle\left( \calM,\epsilon\right)$ }

\State \textbf{Input: } $\mathcal{M}=\{\mm_{i}\}_{i=1}^{d}$ and error parameter $\epsilon$;

\State \textbf{Input: } $\oracle$ a modified one-sided oracle with speed $S=1/2$ and error $\epsilon$;

\State $j=0$, $\ma_0=\mzero$;

\State $\ell_{0}=-\frac{1}{4}$, $u_{0}=\frac{1}{4}$;

\While{$u_j - \ell_j < 1$}

\State $\cp=  (1-2\epsilon)\phideriv(\clb^{-1}\ma_{j}-\ell_{j}\mi)$;

\State $\cn=  (1+2\epsilon)\phideriv(u_{j} \mub-\ma_{j})$;

\State $\Delta_{j}=\oracle\left(\{\mm_{i}\}_{i=1}^{d},\clb \eigenlb,\cp,\cn\right)$;

\State $\ensuremath{\ma_{j+1}\leftarrow\ma_{j}+\epsilon \cdot\Delta_{j}}$;

\State $\delta_{u}=\epsilon \cdot S \cdot \clb \cdot \eigenlb \cdot\frac{1+2\epsilon}{1-4 \epsilon}$
and $\delta_{\ell}= \epsilon \cdot S \cdot \clb \cdot \eigenlb \cdot \frac{1-2\epsilon}{1+4\epsilon}$;

\State $u_{j+1}\leftarrow u_{j}+\delta_{u}$, $\ell_{j+1}\leftarrow\ell_{j}+\delta_{\ell}$;

\State $j\leftarrow j+1.$

\EndWhile

\State \textbf{Return} $\ma_{j}$.
\label{algo1}
\end{algorithmic}
 \end{algorithm}
 
The remaining part of this subsection is dedicated to analyze our algorithm $\TwoSidedOracle$. We start our analysis with \autoref{lem:potential_decreasing} that proves for all iterations $j$ the potential function does not increase and $\lambda_{\min}\left( u_j \mi - \ma_{j} \right)$ and $\lambda_{\min}\left( \gamma^{-1} \ma_{j} - \ell_j \mi \right)$ are both lower bounded by $\eigenlb$. These lower bounds are further used in analysis to show that the addition matrix $\epsilon \cdot \Delta_{j}$ in each iteration $j$ satisfies conditions of \autoref{lem:potentialchange} and \autoref{lem:potential2}.

\begin{lem}
\label{lem:potential_decreasing} In \autoref{algo1} with  $\epsilon \in [0, 1/20]$. Then, for all $j \geq 0$, $ \ell_{j} \mlb \preceq \ma_{j} \preceq u_{j} \mi$ and further it holds that
$\Phi_{u_{j+1},\ell_{j+1}}(\ma_{j+1})\leq\Phi_{u_{j},\ell_{j}}(\ma_{j})$, 
$\lambda_{\min}\left( u_j \mi - \ma_{j} \right)^2 \geq \eigenlb$,
and 
$\lambda_{\min}\left( \gamma^{-1} \ma_{j} - \ell_j \mi \right)^2 \geq \eigenlb$.
\end{lem}

\begin{proof}
For notational convenience, for all $j \geq 0$ define
\[
\phi_j \defeq \Phi_{u_{j},\ell_{j}}(\ma_{j})
\text{ , }
\mm_{j}^{u} \defeq (u_j \mi - \ma_{j})^2
\text{ , and }
\mm_{j}^{\ell} \defeq (\clb^{-1}\ma_{j}-\ell_{j}\mi)^2  ~.
\]
We prove by strong induction that following statements hold for all $j \geq 0$:
\begin{equation}
\label{eq:induct}
\ell_{j}\mlb \preceq \ma_{j} \preceq u_{j} \mi
\text{ , }
\lambda_{\min}(\mm_j^u) \geq \eigenlb
\text{ , }
\lambda_{\min}(\mm_j^\ell) \geq \eigenlb
\text{ , and }
\phi_j \leq \phi_{j - 1} 
\end{equation}
where we trivially define $\ma_{-1} = \ma_{0}$, $u_{-1} = u_0$, and $\ell_{-1} = \ell_{0}$.

For the base case, $j = 0$, we trivially have $\phi_{j} = \phi_{j - 1}$ by our definitions. Further, since $\ma_{0} = \mzero$, $l_0 = -1/4$, $u_0 = 1/4$, and $\gamma \geq 0$ we have that $\ell_{0}\mlb \preceq \ma_{0} \preceq u_{0} \mi$ and $\mm_{0}^{\ell} = \mm_{0}^u = \frac{1}{16} \mi$. Since $1/16 \geq \ln^{-1}(2 e^4 n)$ this completes the proof of \eqref{eq:induct} for $j = 0$. 

Now, suppose \eqref{eq:induct} hold for all $i \in [0, j]$ we show it holds for $j +1$. First we show $\phi_{j+1}\leq \phi_{j}$. Since \eqref{eq:induct} hold for all $i \in [0, j]$ this implies that in each iteration $j$ of the $\TwoSidedOracle$ algorithm, the addition matrix $\epsilon\cdot\Delta_{j}$ with parameter $\delta=\epsilon, \mDelta=\epsilon \Delta_{j}$ and $\mlt= \ma_{j}$ satisfy the conditions of \autoref{lem:potentialchange} for $\alpha = \eigenlb$ because $\Delta_{j} \leq \clb\eigenlb \mi$ and $\lambda_{\min}(\mm_j^u) \geq \eigenlb$ and $\lambda_{\min}(\mm_j^\ell) \geq \eigenlb$ by the inductive hypothesis and we have:
\begin{equation}\label{eq:11}
\begin{split}
\Phi_{u_j,\ell_j}(\ma_j+\epsilon\Delta_j)  \leq  & \Phi_{u_j,\ell_j}(\ma_j) + (1+2\epsilon)\tr\left( \phideriv(u_j\mub-\ma_j)\epsilon\Delta_j\right) -(1-2\epsilon)\tr\left(\phideriv(\clb^{-1}\ma_{j}-\ell_{j} \mi)\epsilon \frac{\Delta_j}{\clb}\right)\\
 &=\Phi_{u_j,\ell_j}(\ma_j) -\left( \tr\left( \cp\epsilon\frac{\Delta_j}{\clb}\right)-\tr\left(\cn\epsilon\Delta_j\right)\right)\\
    & \leq \Phi_{u_j,\ell_j}(\ma_j) - \epsilon \cdot S \cdot \eigenlb \cdot \clb \cdot \left( (1-2\epsilon) \tr(\cp)-(1+2\epsilon)\tr(\cn) \right)~.
 \end{split}
\end{equation}

Similarly parameters $\delta=2\epsilon, \ma=\ma_{j+1}, u=u_j, \ell=\ell_j, \delta_{u}=\delta_{u}, \delta_{l}=\delta_{\ell}$ satisfy the conditions of \autoref{lem:potential2} for $\alpha_{u}=\alpha_{\ell}=(1-\epsilon)^2\eigenlb$ because $\eigenlb \leq \lambda_{\min}(\mm_j^u) $ which further combined with $\epsilon \Delta_{j} \leq \epsilon \clb \eigenlb$ implies 
\[
\lambda_{\min}(u \mi -\ma_{j+1})^{2} =\lambda_{\min}(u \mi -(\ma_{j}+\epsilon \Delta_j))^{2} \geq (1-\epsilon \clb)^2 \eigenlb \geq \alpha_u
\]
Similarly 
\[
\lambda_{\min}(\clb^{-1}\ma_{j+1}-\ell_{j} \mi)^2) \geq \lambda_{\min}(\mm_{j}^{\ell}) \geq  \eigenlb \geq \alpha_{\ell} ~.
\]
Further,   $\epsilon \cdot S \cdot \clb \cdot\frac{1+2\epsilon}{1-4 \epsilon} (1-\epsilon)^{-2} \leq  2\epsilon=\delta$ for $\epsilon \leq 1/20$ and $S\leq 1$ therefore $\delta_{u} \leq \delta \alpha_{u}$ and similarly $\epsilon \cdot S \cdot \clb \cdot\frac{1-2\epsilon}{1+4 \epsilon} (1-\epsilon)^{-2} \leq  2\epsilon$ imply $\delta_{\ell} \leq \delta \alpha_{\ell}$. Now applying 
\autoref{lem:potential2} we get,
\begin{equation}\label{eq:22}
\begin{split}
\Phi_{u_j+\delta_{u},\ell_j+\delta_{\ell}}(\ma_{j+1})  \leq  & \Phi_{u_j,\ell_j}(\ma_{j+1})  - (1-4\epsilon) \delta_{u}\tr\left(\phideriv(u_j\mi-\ma_{j+1})\right)
  +(1+4\epsilon)\delta_{\ell} \tr\left(\phideriv(\clb^{-1}\ma_{j+1}-\ell_{j} \mi)\right)\\
& \leq   \Phi_{u_j,\ell_j}(\ma_{j+1})  - (1-4\epsilon) \delta_{u} \tr\left(\phideriv(u_j\mub-\ma_{j})\right)
  +(1+4\epsilon)\delta_{\ell} \tr\left(\phideriv(\clb^{-1}\ma_{j}-\ell_{j} \mi)\right)\\
& \leq   \Phi_{u_j,\ell_j}(\ma_{j+1})  + \epsilon \cdot S \cdot \clb\cdot \eigenlb \cdot \left( (1-2\epsilon) \tr(\cp)-(1+2\epsilon)\tr(\cn) \right)~.
 \end{split}
\end{equation}
The second inequality follows because $ \ma_{j} \preceq \ma_{j+1}$ and $\ma_{j+1} \preceq u_j \mi$ (because $\epsilon\cdot \Delta_{j} \preceq \epsilon \clb \eigenlb \mi \preceq \eigenlbsq \mi \preceq u_{j}\mi-\ma_{j}$) that further implies:
\[
\tr\left(\phideriv\left(\clb^{-1}\ma_{j+1}-\ell_{j} \mi\right)\right)\leq \tr\left(\phideriv\left(\clb^{-1}\ma_{j}-\ell_{j} \mi\right)\right)\text{ and } \tr\left(\phideriv(u_j\mi-\ma_{j+1})\right) \geq \tr\left(\phideriv(u_j\mub-\ma_{j})\right) ~.
\] In the last inequality we just substituted values for $\delta_{\ell}$, $\delta_{u}$ and wrote in terms of $\cp,\cn$.

\noindent Combining \eqref{eq:11} and \eqref{eq:22} we have $ \phi_{j+1}  \leq \phi_j$. Further, combined with the inductive hypothesis this implies $\phi_{j+1} \leq \phi_0$. However, note that: 
\[
\phi_0=\tr\exp\left(\left(\frac{1}{4}\mi\right)^{-1}\right) + \tr\exp\left(\left(-\frac{-1}{4} \mi\right)^{-1}\right)=2e^4n ~.
\]
Consequently, $\phi_{j+1}=\exp((\mm_{j + 1}^{u})^{-1/2}) +\exp((\mm_{j + 1}^{\ell})^{-1/2}) \leq 2e^4n$ imply $\lambda_{\min}(\mm_{j+1}^{u}) \geq \eigenlb$ and $\lambda_{\min}(\mm_{j+1}^{\ell}) \geq \eigenlb$ respectively. 

Finally it remains to show that $\ell_{j+1}\clb \mi \preceq \ma_{j+1} \preceq u_{j+1}\mi$. We already showed $\ma_{j+1} \preceq u_{j}\mi$ that implies $\ma_{j+1} \preceq u_{j+1}\mi$. For the other side note that $\delta_{\ell}\clb \mi \leq \delta_{\ell}\mi=\epsilon \cdot S \cdot \clb \cdot \eigenlb \cdot \frac{1-2\epsilon}{1+4\epsilon}\mi \leq \clb \cdot \eigenlbsq \mi \leq \clb \cdot \lambda_{\min}(\frac{\ma_{j}-\ell_{j}\clb \mi}{\clb}) \mi\leq \ma_{j}-\ell_{j} \clb \mi$ and this further implies $\ell_{j+1}\clb \mi\leq \ma_{j} \leq \ma_{j+1}$.
\end{proof}

In the lemma above we presented important properties of our potential function and next we present an application of these properties. Our next lemma upper bounds the number of iterations of our algorithm $\TwoSidedOracle$.
\begin{lem}\label{lem:reductionlem}
Let $\epsilon \in [0, 1/20]$. The algorithm
$\TwoSidedOracle\left( \calM,\epsilon \right)$ outputs a $\left(1+O(\epsilon)\right)$-spectral sparsifier in time $O(\alpha \cdot \beta)$, where $\alpha = O(\frac{\log^{2}n}{\epsilon^{2} \cdot \lowerb \cdot S}) $ is the number of iterations of the algorithm and $\beta$ is the cost per iteration of the algorithm.
\end{lem}
\begin{proof}
Note $\delta_{u}-\delta_{\ell} = \Theta \left( \epsilon^{2} \cdot S \cdot \clb \cdot \eigenlb)\right)$ and $u_0-\ell_0=1/2$. Therefore condition $u_j-\ell_j>1$ is achieved in $\Theta\left(\frac{\log^{2}n}{\clb \epsilon^{2} \cdot S}\right)$ many iterations and at termination $u_j =\Omega(\frac{1}{\epsilon})$ because $u_{j}=\frac{1}{4} + j \cdot \delta_{u}$. All that remains to show is that this condition implies that $\TwoSidedOracle\left( \calM,\epsilon \right)$ returns a $\left(1-\epsilon,\clb\right)$-spectral sparsifier, i.e. that 
\[
\frac{u_j}{\ell_j}=\left(1-\frac{u_j-\ell_j}{u_j}\right)^{-1} \leq  1+O(\epsilon)
\]
Therefore, it suffices to show that $\frac{u_j-\ell_j}{u_j}=O(\epsilon)$. However, this follows easily from the termination condition of the algorithm as $\frac{u_j-\ell_j}{u_j}=\frac{u_{j-1}-\ell_{j-1}+ \delta_{u}-\delta_{\ell}}{u_j}$ and $u_{j-1}-\ell_{j-1}<1$, $\delta_{u}-\delta_{\ell}=O(\epsilon^2\clb)$ and $u_j =\Omega(\frac{1}{\epsilon})$.
\end{proof}
The previous lemma upper bounds running time of $\TwoSidedOracle$ in terms of time to run $\oracle$. In the next subsection we reduce our \autoref{prob:recovery} from general case to the identity case and in \autoref{lem:onesided} we give running time to implement $\oracle$ in terms of parameters of this general case that combined with previous lemma proves our main result of this section.

\subsection{Reduction to Identity Case}\label{subsec:reduction}
In this subsection we give an algorithm to solve \autoref{prob:recovery} in a more general setting, that is when matrix $\mb$ is not identity. We do this by applying a very standard trick and reduce the problem of solving general case to the identity case. We describe this standard trick next.

Recall in \autoref{prob:recovery} we are given existence of a vector $w \in \R^d_{\geq 0}$ such that $\lowerb \mb \preceq \sum_{i=1}^{d}w_i \mm'_{i} \preceq  \mb$, and asked to find a vector $w' \in \R^d_{\geq 0}$ such that $(1- O( \epsilon))\lowerb \mb \preceq \sum_{i=1}^{d}w'_i \mm'_{i} \preceq  \mb$. Multiplying $\mb^{-1/2}$ on both sides results in the following equivalent formulation: Given existence of a vector $w \in \R^d_{\geq 0}$ such that $\lowerb \mi \preceq \sum_{i=1}^{d}w_i \mb^{-1/2}\mm'_{i}\mb^{-1/2} \preceq  \mi$, find a vector $w' \in \R^d_{\geq 0}$ such that $(1- O( \epsilon ) )\lowerb \mi \preceq \sum_{i=1}^{d}w'_i \mb^{-1/2}\mm'_{i}\mb^{-1/2} \preceq  \mi$. This new formulation falls under identity setting of \autoref{subsec:twosided}  by choosing $\mm_{i}=\mb^{-1/2}\mm'_{i}\mb^{-1/2}$ and consequently, by \autoref{lem:reductionlem} algorithm $\TwoSidedOracle$ solves this problem in time $O\left(\frac{\log^{2}n}{\epsilon^{2} \cdot \lowerb \cdot S} \cdot \beta \right)$, where $\beta$ is time per iteration. In next several  lemmas we upper bound the time per iteration for this general case. For these results we make critical use of the following result.

\begin{thm}
[\cite{Allen-ZhuLO16}]\label{thm:solve_SDP} Given a SDP 
\[
\mathsf{OPT} = \max_{x\geq0}c^{\rot}x\text{ subject to }\sum_{i=1}^{m}x_{i}\ma_{i}\preceq\mb
\]
with $\ma_{i}\succeq\mzero$, $\mb\succeq\mzero$ and $c\in\R^{m}$. Suppose that we are given a direct access to the vector $c\in\Rset^m$ and an indirect access to $\ma_{i}$ and $\mb$ via an oracle $\mathcal{O}_{L,\delta}$ which inputs a vector $x\in\R^{m}$ and outputs a vector $v\in\R^{m}$ such that
\[
v_{i}\in\left(1\pm\frac{\delta}{2}\right)\left[\ma_{i}\bullet\mb^{-1/2}\exp\left(L\cdot\mb^{-1/2}\left(\sum_{i}x_{i}\ma_{i}-\mb\right)\mb^{-1/2}\right)\mb^{-1/2}\right]
\]
in $\mathcal{W}_{L,\delta}$ time for any $x$ such that $x_i\geq0$
and $\sum_{i=1}^{m}x_{i}\ma_{i}\preceq2\mb$. Then, we can output
$x$ such that 
\[
\E\left[c^\top x \right]\geq(1-O(\delta))\mathsf{OPT}\quad\text{with}\quad\sum_{i=1}^{m}x_{i}\ma_{i}\preceq\mb
\]
 in $O\left(\mathcal{W}_{L,\delta}\log m\cdot\log\left(nm/\delta\right)/\delta^{3}\right)$ time where $L=(4/\delta)\cdot\log(nm/\delta)$.
\end{thm}

\begin{lemma}\label{lem:onesided}
$\oracle$ with speed $\speed$ and error $\epsilon$ can be implemented in time 
\[
\rt ~.
\]
\end{lemma}
\begin{proof}
Our proof follows along the similar lines as in \cite{ls17}. An implementation for $\oracle$ with speed $S=1/2$ and error $\epsilon$ is equivalent to approximately solving the following optimization problem 
\begin{equation}\label{eq:SDP}
\max_{\alpha_i\geq0}\cs\bullet\left(\sum_{i=1}^{d}\alpha_{i}\mm_{i}\right)\text{ subject to }\Delta=\sum_{i=1}^{d}\alpha_{i}\mm_{i}\preceq \constb \mi
\end{equation}
where $\cs=\frac{1}{\clb}\cp - \cn$ and we recall definitions of $\cp$ and $\cn$ below:
$$\cp=  (1-2\epsilon)\phideriv(\clb^{-1}\ma_{j}-\ell_{j} \mi) \text{ and } \cn=  (1+2\epsilon)\phideriv(u_{j}\mi-\ma_{j})~.$$ 
We show how to invoke the packing SDP result from \autoref{thm:solve_SDP} to solve it faster.

Let $\error=\epsilon/4$. Suppose for each iteration $j$ the following two conditions hold: \\
\\
{\bf (1)} we have access to $c_i'$ and $d_{i}'$ that are multiplicative $(1\pm \error)$ approximation to $c_i=\cp\bullet \mm_{i} $ and $d_i=\cn\bullet \mm_{i} $ respectively, then solve optimization problem \eqref{eq:SDP} with parameter $\mb=\constb \mi$ (constraint matrix) and objective value $\sum_{i=1}^{d}\alpha_{i}((1-\error)\frac{c_{i}'}{\clb}-(1+\error)d_{i}')$ instead of the original $\sum_{i=1}^{d}\alpha_{i}(\frac{c_{i}}{\clb}-d_{i})$, then $\Delta=\constb \cdot \sum_{i=1}^{d}w_i\mm_{i}$ is a feasible solution with objective value greater than $\constb \left[ (1-\error)^2\tr (\cp) - (1+\error)^2\tr(\cn)\right]$ and the optimum solution to optimization problem \eqref{eq:SDP} with $c_i'$ and $d_{i}'$ parameters is 
\[
\mathsf{OPT}'  \geq \constb \left[ (1-\error)^2\tr (\cp) - (1+\error)^2\tr(\cn)\right] \geq \constb \left[ (1-\epsilon)\tr (\cp) - (1+\epsilon)\tr(\cn)\right]
\]\\
\\
{\bf (2)} we have access to vector $v\in\R^{d}$ such that: 
\[
v_{i}\in\left(1\pm\frac{\delta}{2}\right)\left[\mm_{i}\bullet\constb\exp\left(L\cdot\constb\left(\sum_{i}x_{i}\mm_{i}-\constb\mi\right)\right)\right]
\]
for any $x$ such that $x_i\geq0$ and $\sum_{i=1}^{d}x_{i}\mm_{i}\preceq2\constb\mi$. 

Under these two conditions we can use \autoref{thm:solve_SDP} and find vector $\alpha \in \R_{\geq 0}^{d}$ such that $\sum_{i=1}^{d}\alpha_{i}\mm_{i} \preceq\constb\mi$ and
\begin{align*}
\cs\bullet \left[\sum_{i \in [d]}\alpha_{i}\mm_{i} \right] & =\sum_{i \in [d]}(\frac{c_i}{\clb}-d_{i}) \alpha_{i} \geq \sum_{i \in [d]}((1-\error)\frac{c_i'}{\clb}-(1+\error)d_{i}') \alpha_{i}\\
& \geq (1-O(\delta)) \mathsf{OPT}'  \qquad (\text{By Theorem } \ref{thm:solve_SDP})\\
& \geq \frac{1}{2}\cdot \constb \left[ (1-\epsilon)\tr (\cp) - (1+\epsilon)\tr(\cn)\right]
\end{align*}
In the last step we chose $O(\delta) =1- S=1/2$. Consequently, proving this theorem  boils down to approximating $c_i$ and $d_{i}$ to multiplicative $(1\pm\error)$ and $v_{i}$ to multiplicative $\Theta(1)$. Recall our matrix $\mm_{i}= \mb^{-1/2}\mm'_{i}\mb^{-1/2}$ and we use a standard and slightly general reasoning of \cite{ls17, SpielmanS08, ALO15, LeeS15a, PengS14, KyngLPSS16, NN13, LiMP13, CohenLMMPS14} to approximate these quantities by constructing a Johnson-Lindenstrauss (JL) sketch. In \autoref{sec:JL} we show that the complete vectors $c$, $d$ and $v$ can be written down in time $\rt$ in \autoref{lem:JL} completing the proof.
\end{proof}
\autoref{lem:onesided} gives an upper bound on running time of $\oracle$ and combining it with \autoref{lem:reductionlem} proves our main \autoref{thm:recmain}.

\section{Acknowledgments}

We thank Moses Charikar and Yin Tat Lee for helpful discussions.

\bibliographystyle{plain}
\bibliography{recovery}

\begin{thebibliography}{10}

\bibitem{ajss18}
AmirMahdi Ahmadinejad, Arun Jambulapati, Amin Saberi, and Aaron Sidford.
\newblock Perron-frobenius theory in nearly linear time: Positive eigenvectors,
  m-matrices, graph kernels, and other applications.
\newblock {\em CoRR}, abs/1810.02348, 2018.

\bibitem{Allen-ZhuLO16}
Zeyuan {Allen Zhu}, Yin~Tat Lee, and Lorenzo Orecchia.
\newblock Using optimization to obtain a width-independent, parallel, simpler,
  and faster positive {SDP} solver.
\newblock In {\em Proceedings of the Twenty-Seventh Annual {ACM-SIAM} Symposium
  on Discrete Algorithms, {SODA} 2016, Arlington, VA, USA, January 10-12,
  2016}, pages 1824--1831, 2016.

\bibitem{ALO15}
Zeyuan {Allen Zhu}, Zhenyu Liao, and Lorenzo Orecchia.
\newblock Spectral sparsification and regret minimization beyond matrix
  multiplicative updates.
\newblock {\em CoRR}, abs/1506.04838, 2015.

\bibitem{AndoniCKQWZ16}
Alexandr Andoni, Jiecao Chen, Robert Krauthgamer, Bo~Qin, David~P. Woodruff,
  and Qin Zhang.
\newblock On sketching quadratic forms.
\newblock In {\em Proceedings of the 2016 {ACM} Conference on Innovations in
  Theoretical Computer Science, Cambridge, MA, USA, January 14-16, 2016}, pages
  311--319, 2016.

\bibitem{BomanHV08}
Erik~G. Boman, Bruce Hendrickson, and Stephen~A. Vavasis.
\newblock Solving elliptic finite element systems in near-linear time with
  support preconditioners.
\newblock {\em {SIAM} J. Numerical Analysis}, 46(6):3264--3284, 2008.

\bibitem{CCLPT14}
Dehua Cheng, Yu~Cheng, Yan Liu, Richard Peng, and Shang{-}Hua Teng.
\newblock Scalable parallel factorizations of {SDD} matrices and efficient
  sampling for gaussian graphical models.
\newblock {\em CoRR}, abs/1410.5392, 2014.

\bibitem{CG18}
Yu~Cheng and Rong Ge.
\newblock Non-convex matrix completion against a semi-random adversary.
\newblock {\em CoRR}, abs/1803.10846, 2018.

\bibitem{ChristianoKMST11}
Paul Christiano, Jonathan~A. Kelner, Aleksander Madry, Daniel~A. Spielman, and
  Shang{-}Hua Teng.
\newblock Electrical flows, laplacian systems, and faster approximation of
  maximum flow in undirected graphs.
\newblock In {\em Proceedings of the 43rd {ACM} Symposium on Theory of
  Computing, {STOC} 2011, San Jose, CA, USA, 6-8 June 2011}, pages 273--282,
  2011.

\bibitem{CohenKPPRSV17}
Michael~B. Cohen, Jonathan~A. Kelner, John Peebles, Richard Peng, Anup~B. Rao,
  Aaron Sidford, and Adrian Vladu.
\newblock Almost-linear-time algorithms for markov chains and new spectral
  primitives for directed graphs.
\newblock In {\em Proceedings of the 49th Annual {ACM} {SIGACT} Symposium on
  Theory of Computing, {STOC} 2017, Montreal, QC, Canada, June 19-23, 2017},
  pages 410--419, 2017.

\bibitem{CohenKPPSV16}
Michael~B. Cohen, Jonathan~A. Kelner, John Peebles, Richard Peng, Aaron
  Sidford, and Adrian Vladu.
\newblock Faster algorithms for computing the stationary distribution,
  simulating random walks, and more.
\newblock In {\em {IEEE} 57th Annual Symposium on Foundations of Computer
  Science, {FOCS} 2016, 9-11 October 2016, Hyatt Regency, New Brunswick, New
  Jersey, {USA}}, pages 583--592, 2016.

\bibitem{CohenKMPPRX14}
Michael~B. Cohen, Rasmus Kyng, Gary~L. Miller, Jakub~W. Pachocki, Richard Peng,
  Anup~B. Rao, and Shen~Chen Xu.
\newblock Solving {SDD} linear systems in nearly
  \emph{m}log\({}^{\mbox{1/2}}\)\emph{n} time.
\newblock In {\em Symposium on Theory of Computing, {STOC} 2014, New York, NY,
  USA, May 31 - June 03, 2014}, pages 343--352, 2014.

\bibitem{CohenLMMPS14}
Michael~B. Cohen, Yin~Tat Lee, Cameron Musco, Christopher Musco, Richard Peng,
  and Aaron Sidford.
\newblock Uniform sampling for matrix approximation.
\newblock {\em CoRR}, abs/1408.5099, 2014.

\bibitem{CohenMSV17}
Michael~B. Cohen, Aleksander Madry, Piotr Sankowski, and Adrian Vladu.
\newblock Negative-weight shortest paths and unit capacity minimum cost flow in
  {\~{o}} (\emph{m}\({}^{\mbox{10/7}}\) log \emph{W}) time (extended abstract).
\newblock In {\em Proceedings of the Twenty-Eighth Annual {ACM-SIAM} Symposium
  on Discrete Algorithms, {SODA} 2017, Barcelona, Spain, Hotel Porta Fira,
  January 16-19}, pages 752--771, 2017.

\bibitem{CohenMTV17}
Michael~B. Cohen, Aleksander Madry, Dimitris Tsipras, and Adrian Vladu.
\newblock Matrix scaling and balancing via box constrained newton's method and
  interior point methods.
\newblock In {\em 58th {IEEE} Annual Symposium on Foundations of Computer
  Science, {FOCS} 2017, Berkeley, CA, USA, October 15-17, 2017}, pages
  902--913, 2017.

\bibitem{DaitchS08}
Samuel~I. Daitch and Daniel~A. Spielman.
\newblock Faster approximate lossy generalized flow via interior point
  algorithms.
\newblock In {\em Proceedings of the 40th Annual {ACM} Symposium on Theory of
  Computing, Victoria, British Columbia, Canada, May 17-20, 2008}, pages
  451--460, 2008.

\bibitem{DurfeeKPRS17}
David Durfee, Rasmus Kyng, John Peebles, Anup~B. Rao, and Sushant Sachdeva.
\newblock Sampling random spanning trees faster than matrix multiplication.
\newblock In {\em Proceedings of the 49th Annual {ACM} {SIGACT} Symposium on
  Theory of Computing, {STOC} 2017, Montreal, QC, Canada, June 19-23, 2017},
  pages 730--742, 2017.

\bibitem{fallat2017}
Shaun Fallat, Steffen Lauritzen, Kayvan Sadeghi, Caroline Uhler, Nanny Wermuth,
  and Piotr Zwiernik.
\newblock Total positivity in markov structures.
\newblock {\em Ann. Statist.}, 45(3):1152--1184, 06 2017.

\bibitem{effResRecArx18}
Jeremy~G. Hoskins, Cameron Musco, Christopher Musco, and Charalampos~E.
  Tsourakakis.
\newblock Learning networks from random walk-based node similarities.
\newblock {\em CoRR}, abs/1801.07386, 2018.

\bibitem{JY11}
Rahul Jain and Penghui Yao.
\newblock A parallel approximation algorithm for positive semidefinite
  programming.
\newblock In {\em Proceedings of the 2011 IEEE 52Nd Annual Symposium on
  Foundations of Computer Science}, FOCS '11, pages 463--471, Washington, DC,
  USA, 2011. IEEE Computer Society.

\bibitem{JY12}
Rahul Jain and Penghui Yao.
\newblock A parallel approximation algorithm for mixed packing and covering
  semidefinite programs.
\newblock {\em CoRR}, abs/1201.6090, 2012.

\bibitem{KapralovLMMS14}
Michael Kapralov, Yin~Tat Lee, Cameron Musco, Christopher Musco, and Aaron
  Sidford.
\newblock Single pass spectral sparsification in dynamic streams.
\newblock In {\em 55th {IEEE} Annual Symposium on Foundations of Computer
  Science, {FOCS} 2014, Philadelphia, PA, USA, October 18-21, 2014}, pages
  561--570, 2014.

\bibitem{KARLIN1983419}
Samuel Karlin and Yosef Rinott.
\newblock M-matrices as covariance matrices of multinormal distributions.
\newblock {\em Linear Algebra and its Applications}, 52-53:419 -- 438, 1983.

\bibitem{KLOS14}
Jonathan~A. Kelner, Yin~Tat Lee, Lorenzo Orecchia, and Aaron Sidford.
\newblock An almost-linear-time algorithm for approximate max flow in
  undirected graphs, and its multicommodity generalizations.
\newblock In {\em Proceedings of the Twenty-Fifth Annual {ACM-SIAM} Symposium
  on Discrete Algorithms, {SODA} 2014, Portland, Oregon, USA, January 5-7,
  2014}, pages 217--226, 2014.

\bibitem{KelnerM09}
Jonathan~A. Kelner and Aleksander Madry.
\newblock Faster generation of random spanning trees.
\newblock In {\em 50th Annual {IEEE} Symposium on Foundations of Computer
  Science, {FOCS} 2009, October 25-27, 2009, Atlanta, Georgia, {USA}}, pages
  13--21, 2009.

\bibitem{KOSZ13}
Jonathan~A. Kelner, Lorenzo Orecchia, Aaron Sidford, and Zeyuan {Allen Zhu}.
\newblock A simple, combinatorial algorithm for solving {SDD} systems in
  nearly-linear time.
\newblock In {\em Symposium on Theory of Computing Conference, STOC'13, Palo
  Alto, CA, USA, June 1-4, 2013}, pages 911--920, 2013.

\bibitem{KoutisMP10}
Ioannis Koutis, Gary~L. Miller, and Richard Peng.
\newblock Approaching optimality for solving {SDD} linear systems.
\newblock In {\em 51th Annual {IEEE} Symposium on Foundations of Computer
  Science, {FOCS} 2010, October 23-26, 2010, Las Vegas, Nevada, {USA}}, pages
  235--244, 2010.

\bibitem{KoutisMP11}
Ioannis Koutis, Gary~L. Miller, and Richard Peng.
\newblock A nearly-m log n time solver for {SDD} linear systems.
\newblock In {\em {IEEE} 52nd Annual Symposium on Foundations of Computer
  Science, {FOCS} 2011, Palm Springs, CA, USA, October 22-25, 2011}, pages
  590--598, 2011.

\bibitem{KyngLPSS16}
Rasmus Kyng, Yin~Tat Lee, Richard Peng, Sushant Sachdeva, and Daniel~A.
  Spielman.
\newblock Sparsified cholesky and multigrid solvers for connection laplacians.
\newblock In {\em Proceedings of the 48th Annual {ACM} {SIGACT} Symposium on
  Theory of Computing, {STOC} 2016, Cambridge, MA, USA, June 18-21, 2016},
  pages 842--850, 2016.

\bibitem{KyngRSS15}
Rasmus Kyng, Anup Rao, Sushant Sachdeva, and Daniel~A. Spielman.
\newblock Algorithms for lipschitz learning on graphs.
\newblock In {\em Proceedings of The 28th Conference on Learning Theory, {COLT}
  2015, Paris, France, July 3-6, 2015}, pages 1190--1223, 2015.

\bibitem{KyngS16}
Rasmus Kyng and Sushant Sachdeva.
\newblock Approximate gaussian elimination for laplacians - fast, sparse, and
  simple.
\newblock In {\em {IEEE} 57th Annual Symposium on Foundations of Computer
  Science, {FOCS} 2016, 9-11 October 2016, Hyatt Regency, New Brunswick, New
  Jersey, {USA}}, pages 573--582, 2016.

\bibitem{LeeRS13}
Yin~Tat Lee, Satish Rao, and Nikhil Srivastava.
\newblock A new approach to computing maximum flows using electrical flows.
\newblock In {\em Symposium on Theory of Computing Conference, STOC'13, Palo
  Alto, CA, USA, June 1-4, 2013}, pages 755--764, 2013.

\bibitem{LS13}
Yin~Tat Lee and Aaron Sidford.
\newblock Efficient accelerated coordinate descent methods and faster
  algorithms for solving linear systems.
\newblock In {\em 54th Annual {IEEE} Symposium on Foundations of Computer
  Science, {FOCS} 2013, 26-29 October, 2013, Berkeley, CA, {USA}}, pages
  147--156, 2013.

\bibitem{LS14}
Yin~Tat Lee and Aaron Sidford.
\newblock Path finding methods for linear programming: Solving linear programs
  in {\~{o}}(vrank) iterations and faster algorithms for maximum flow.
\newblock In {\em 55th {IEEE} Annual Symposium on Foundations of Computer
  Science, {FOCS} 2014, Philadelphia, PA, USA, October 18-21, 2014}, pages
  424--433, 2014.

\bibitem{LeeS15a}
Yin~Tat Lee and He~Sun.
\newblock Constructing linear-sized spectral sparsification in almost-linear
  time.
\newblock In {\em {IEEE} 56th Annual Symposium on Foundations of Computer
  Science, {FOCS} 2015, Berkeley, CA, USA, 17-20 October, 2015}, pages
  250--269, 2015.

\bibitem{ls17}
Yin~Tat Lee and He~Sun.
\newblock An sdp-based algorithm for linear-sized spectral sparsification.
\newblock In {\em Proceedings of the 49th Annual {ACM} {SIGACT} Symposium on
  Theory of Computing, {STOC} 2017, Montreal, QC, Canada, June 19-23, 2017},
  pages 678--687, 2017.

\bibitem{LiMP13}
Mu~Li, Gary~L. Miller, and Richard Peng.
\newblock Iterative row sampling.
\newblock In {\em 54th Annual {IEEE} Symposium on Foundations of Computer
  Science, {FOCS} 2013, 26-29 October, 2013, Berkeley, CA, {USA}}, pages
  127--136, 2013.

\bibitem{Madry16}
Aleksander Madry.
\newblock Computing maximum flow with augmenting electrical flows.
\newblock In {\em {IEEE} 57th Annual Symposium on Foundations of Computer
  Science, {FOCS} 2016, 9-11 October 2016, Hyatt Regency, New Brunswick, New
  Jersey, {USA}}, pages 593--602, 2016.

\bibitem{NN13}
Jelani Nelson and Huy~L. Nguyen.
\newblock {OSNAP:} faster numerical linear algebra algorithms via sparser
  subspace embeddings.
\newblock {\em CoRR}, abs/1211.1002, 2012.

\bibitem{OrecchiaSV12}
Lorenzo Orecchia, Sushant Sachdeva, and Nisheeth~K. Vishnoi.
\newblock Approximating the exponential, the lanczos method and an
  {\~{o}}(\emph{m})-time spectral algorithm for balanced separator.
\newblock In {\em Proceedings of the 44th Symposium on Theory of Computing
  Conference, {STOC} 2012, New York, NY, USA, May 19 - 22, 2012}, pages
  1141--1160, 2012.

\bibitem{OrecchiaV11}
Lorenzo Orecchia and Nisheeth~K. Vishnoi.
\newblock Towards an sdp-based approach to spectral methods: {A}
  nearly-linear-time algorithm for graph partitioning and decomposition.
\newblock In {\em Proceedings of the Twenty-Second Annual {ACM-SIAM} Symposium
  on Discrete Algorithms, {SODA} 2011, San Francisco, California, USA, January
  23-25, 2011}, pages 532--545, 2011.

\bibitem{PengS14}
Richard Peng and Daniel~A. Spielman.
\newblock An efficient parallel solver for {SDD} linear systems.
\newblock In {\em Symposium on Theory of Computing, {STOC} 2014, New York, NY,
  USA, May 31 - June 03, 2014}, pages 333--342, 2014.

\bibitem{PT12}
Richard Peng and Kanat Tangwongsan.
\newblock Faster and simpler width-independent parallel algorithms for positive
  semidefinite programming.
\newblock {\em CoRR}, abs/1201.5135, 2012.

\bibitem{Schild18}
Aaron Schild.
\newblock An almost-linear time algorithm for uniform random spanning tree
  generation.
\newblock In {\em Proceedings of the 50th Annual {ACM} {SIGACT} Symposium on
  Theory of Computing, {STOC} 2018, Los Angeles, CA, USA, June 25-29, 2018},
  pages 214--227, 2018.

\bibitem{mm14}
Martin Slawski and Matthias Hein.
\newblock Estimation of positive definite m-matrices and structure learning for
  attractive gaussian markov random fields.
\newblock 473, 04 2014.

\bibitem{SpielmanS08}
Daniel~A. Spielman and Nikhil Srivastava.
\newblock Graph sparsification by effective resistances.
\newblock In {\em Proceedings of the 40th Annual {ACM} Symposium on Theory of
  Computing, Victoria, British Columbia, Canada, May 17-20, 2008}, pages
  563--568, 2008.

\bibitem{SpielmanT04}
Daniel~A. Spielman and Shang{-}Hua Teng.
\newblock Nearly-linear time algorithms for graph partitioning, graph
  sparsification, and solving linear systems.
\newblock In {\em Proceedings of the 36th Annual {ACM} Symposium on Theory of
  Computing, Chicago, IL, USA, June 13-16, 2004}, pages 81--90, 2004.

\bibitem{Teng10}
Shang{-}Hua Teng.
\newblock The laplacian paradigm: Emerging algorithms for massive graphs.
\newblock In {\em Theory and Applications of Models of Computation, 7th Annual
  Conference, {TAMC} 2010, Prague, Czech Republic, June 7-11, 2010.
  Proceedings}, pages 2--14, 2010.

\bibitem{Vishnoi13}
Nisheeth~K. Vishnoi.
\newblock Lx = b.
\newblock {\em Foundations and Trends in Theoretical Computer Science},
  8(1-2):1--141, 2013.

\bibitem{Williams12}
Virginia~Vassilevska Williams.
\newblock Multiplying matrices faster than coppersmith-winograd.
\newblock In {\em Proceedings of the 44th Symposium on Theory of Computing
  Conference, {STOC} 2012, New York, NY, USA, May 19 - 22, 2012}, pages
  887--898, 2012.

\bibitem{Young14}
Neal~E. Young.
\newblock Nearly linear-time approximation schemes for mixed packing/covering
  and facility-location linear programs.
\newblock {\em CoRR}, abs/1407.3015, 2014.

\end{thebibliography}

\appendix
\section{Approximation to Matrix Square Root}
\label{sec:square_root}

In this section, we provide approximation to square root of a matrix. This result was already known in \cite{CCLPT14} and we use the same proof to prove a slightly general theorem. We start by restating a lemma from \cite{CCLPT14} and in our next result we give polynomial approximation result to square root of a matrix.
\begin{restatable}[Lemma 4.1 in \cite{CCLPT14}]{lemma}{restateTaylor}
\label{lem:taylorscalar}
 Fix $p \in [-1, 1]$ and $\delta \in (0, 1)$, for any error tolerance $\epsilon > 0$, there exists a $t$-degree polynomial $T_{p,t} : \R \rightarrow \R$ with $t \le \frac{\log(1/(\epsilon (1-\delta)^2)}{1-\delta}$, such that for all $\lambda \in [1-\delta, 1+\delta]$,
\begin{align}
\exp(-\epsilon) \lambda^{p} \leq T_{p,t}(\lambda) \leq  \exp(\epsilon) \lambda^{p}.
\end{align}
\end{restatable}

\lemprecond*

\begin{proof}
Note $\alpha \M{Z} \M{Z}^{\top} \preceq \MM^{-1} \preceq \M{Z} \M{Z}^{\top}$ implies $ (\M{Z} \M{Z}^{\top})^{-1} \preceq \MM \preceq1/\alpha (\M{Z} \M{Z}^{\top})^{-1}$, further multiplying $\M{Z}^{\top}$ by ;eft and $\M{Z}$ by right we get $\M{I} \preceq \M{Z}^{\top} \MM \M{Z} \preceq 1/\alpha \M{I}$.
This implies that $\kappa(\M{Z}^{\top} \MM \M{Z}) = \alpha$, which we can scale $\M{Z}^{\top} \MM \M{Z}$ so that its eigenvalues lie in $[1-\delta,1+\delta]$ for $\delta=1-\alpha>0$.

By applying \autoref{lem:taylorscalar} on its eigenvalues, there is an $O({\frac{1}{\alpha}\log(1/(\epsilon\alpha))})$ degree polynomial $T_{-\Half, t}(\cdot)$ that approximates the inverse square root of $\M{Z}^{\top} \MM \M{Z}$, that is,
\begin{align}
  \left(T_{-\Half, \OO{\log(1/\epsilon)}} \left(\M{Z}^{\top} \MM \M{Z}\right)\right)^2 \approx_\epsilon \left(\M{Z}^{\top} \MM \M{Z}\right)^{-1}.
\end{align}
Here, we can rescale $\M{Z}^{\top} \MM \M{Z}$ back inside $T_{-\Half, t}(\cdot)$, which does not affect the multiplicative error. Then we have,
\begin{align}
  \M{Z} \left(T_{-\Half, \OO{\log(1/\epsilon)}}\left(\M{Z}^{\top} \MM \M{Z}\right)\right)^2 \M{Z}^{\top} \approx_\epsilon \M{Z} \left(\M{Z}^{\top} \MM \M{Z}\right)^{-1} \M{Z}^{\top} = \MM^{-1}.
\end{align}
So if we define the linear operator
  $\tilde{\M{C}} = \M{Z} \left(T_{-\Half, \OO{\log(1/\epsilon)}}\left(\M{Z}^{\top} \MM \M{Z}\right) \right)$,
  $\tilde{\M{C}}$ satisfies the claimed properties.
\end{proof}
\newcommand{\mpp}{\textbf{P}}
\newcommand{\me}{\textbf{E}}
\newcommand{\mq}{\textbf{Q}}
\newcommand{\mda}{\textbf{D}_{\ma}}
\newcommand{\mdb}{\textbf{D}_{\mb}}
\renewcommand{\rt}{\tilde{O}\left(\frac{1}{\gamma^{4}\error^{O(1)}}\left(\runtime_{\mb} + \runtime_{\mb^{-1}} + \runtime_{MV} + \runtime_{QF} + \runtime_{SQ}\right) \right)}
\section{Fast Implementation of $\oracle$}\label{sec:JL}
In this section we show how to approximate vector $c \in \R^d$ used in subroutine $\oracle$ and a similar analysis can be done to approximate vector $v \in \R^d$ as well. The main theme of this section is to use the famous Johnson Lindenstrauss (JL) sketch to compute these vectors quickly. Below is our main lemma of this section; the proof is adapted from \cite{ls17}.

\begin{lemma}\label{lem:JL}
For any iteration $j$ of \autoref{algo1}, the vector $c$, $d$ and $v$ needed for subroutine $\oracle$ can be computed in time $\rt$, where $\error=\frac{\epsilon}{4}$.
\end{lemma}
In the remaining part of this section we prove this lemma. Recall at each coordinate $i \in [d]$, our vectors $c$ and $d$ take values $c_i=\cp\bullet \mm_{i} $ and $d_i=\cn\bullet \mm_{i} $ respectively and approximating them within multiplicative $(1\pm \error)$ error suffices to implement our subroutine $\oracle$. Below we give a more general formulation which captures the problem of finding these multiplicative $(1\pm \error)$ estimates of $\cp \bullet \mm_{i}  $ and $\cn \bullet \mm_{i}  $ as special cases: For each iteration $j$ we wish to
\begin{itemize}
\item find $(1\pm \error)$ multiplicative approximation to $\tr(\mx \mb^{-1/2} \phideriv(u\mi-\ma) \mb^{-1/2})$ for any $\mx \succeq \mzero$ and
\item find $(1\pm \error)$ multiplicative approximation to $\tr(\mx \mb^{-1/2} \phideriv\left(\clb^{-1}\ma - \ell \mi\right) \mb^{-1/2})$ for any $\mx \succeq \mzero$.
\end{itemize}
In the formula above we discarded the subscripts with respect to $j$ for $\ma_{j},u_j$ and $\ell_{j}$. Here we are using notations from \autoref{subsec:reduction} where $\mm_{i}=\mb^{-1/2}\mm'_{i}\mb^{-1/2}$ and each of these $\mm'_{i} \in \mathcal{M}$. In this notation  our matrix $\ma$ is of the form $\ma = \sum_{i=1}^{d}\alpha_i \mm_{i}$ for some vector $\alpha \in \R^d_{\geq 0}$ and we define $\hat{\mm} \defeq \sum_{i=1}^{d}\alpha_i \mm'_{i}$. Further our subroutine $\oracle$ guarantees two basic properties for matrix $\ma$: 
\begin{itemize}
	\item $\ell \clb \mi \preceq \ma \preceq u\mi$ and 
	\item $\lambda_{\min} (u\mi-\ma)$, $\lambda_{\min} (\clb^{-1}\ma-\ell\mi)$ are both $\Omega(\log^{-1} n)$. 
\end{itemize}
Now by the Lowner-Heinz inequality, we get
\[
\Omega(\log^{-1} n)\mi \preceq (u\mi-\ma) \preceq (u\mi - \ell\clb \mi) \preceq\frac{1}{\clb}(u\clb \mi - \ell \mi)\leq \frac{1}{\clb}(u\mi - \ell \mi)\leq\frac{1}{\clb}  \mi
\]
and
\[
\Omega(\log^{-1} n)\mi \preceq (\clb^{-1}\ma-\ell \mi) \preceq \frac{1}{\clb}(u\mi - \ell\clb  \mi) \preceq \frac{1}{\clb^2} \mi
\]
In the above inequalities we used $u - \ell \leq 1$ that is guaranteed by our algorithm $\TwoSidedOracle$.
Before we proceed to our result we first prove two important and technical lemmas that will be very crucial to our future analysis. The first lemma is a basic matrix inequality that will help prove our later lemma. The second lemma talks about approximating function $x^{-1} \exp(x^{-1})$ by a low degree polynomial that in general are much nicer to work with.
\begin{lemma}
\label{lem:commute}
For all integers $i \geq 0$. If $\mzero \preceq \ma \preceq \mb$ and $\ma\mb = \mb \ma$, then $\ma^{i} \preceq \mb^{i}$. 
\end{lemma}

\begin{proof} We know that a set of matrices commute iff they are simultaneously 
diagonalizable. Thus there exists $\mpp$ such that $\mpp^{-1} \ma \mpp = \mda$ and $\mpp^{-1} \mb \mpp = \mdb$. Since $\ma$ and $\mb$ are symmetric, $\mpp$ must be orthogonal. 
Thus, $\mzero \preceq \ma \preceq \mb$ implies $\mzero \preceq \mda \preceq \mdb$. This clearly implies $\mda^{i} \preceq \mdb^{i}$ and we have $\ma^{i} = \mpp \mda \mpp^{-1} \mpp \mda \mpp^{-1}\dots \mpp \mda \mpp^{-1}  = \mpp \mda^i \mpp^{-1} \preceq \mpp \mdb^{i} \mpp^{-1} = \mpp \mdb \mpp^{-1} \mpp \mdb \mpp^{-1}\dots \mpp \mdb \mpp^{-1} = \mb^i$.
\end{proof}
Note that in general $\ma \preceq \mb$ does not imply $\ma^2 \preceq \mb^2$. However,  the previous lemma shows that if the matrices are PSD and commute then $\ma \preceq \mb$ does imply $\ma^2 \preceq \mb^2$. With this in mind, we provide our next technical result \autoref{thm:taylor_expand} using a standard result \autoref{thm:cauchy}.

\begin{theorem}\label{thm:cauchy}
[\textbf{Cauchy's Estimates}]\label{thm:cauchy_estimate} Suppose $f$ is holomorphic
on a neighborhood of the ball $B(s)\triangleq\{z\in\mathbb{C}\ :\ \left|z-s\right|\leq r\}$,
then we have that
\[
\left|f^{(k)}(s)\right|\leq\frac{k!}{r^{k}}\sup_{z\in B(s)}\left|f(z)\right|.
\]
\end{theorem}

\begin{theorem}\label{thm:taylor_expand}
Let $t\geq 1$ and $0 < x \leq t$, and set $f(x) = x^{-1} \exp(x^{-1})$. Then,
\[
\left|f(x)-\sum_{k=0}^{d}\frac{1}{k!}f^{(k)}(1)(x-1)^{k}\right|\leq 4t(d+1)\cdot \mathrm{e}^{\frac{2t-1}{x}-\log x-x(d+1)/t} f(x).
\]
\end{theorem}
\begin{proof}
By the formula of the remainder term in Taylor series, we have that
\[
f(x)=\sum_{k=0}^{d}\frac{1}{k!}f^{(k)}(t)(x-t)^{k}+\frac{1}{d!}\int_{t}^{x}f^{(d+1)}(s)(x-s)^{d}ds.
\]
For any $s\in[x,t]$, we define $D(s)=\{z\in\mathbb{C}\ :\ \left|z-s\right|\leq s-\frac{x}{2t}\}$. Note that since $f(z)$ is decreasing for all $z > 0$,  $\left|f(z)\right|\leq(x/2t)^{-1}\exp(2t/x)$ on $z\in D(s)$ because $z\geq\frac{x}{2t}$.
Cauchy's remainder theorem now shows that
\[
\left|f^{(d+1)}(s)\right|\leq\frac{(d+1)!}{(s-\frac{x}{2t})^{d+1}}\sup_{z\in B(s)}\left|f(z)\right|\leq\frac{(d+1)!}{(s-\frac{x}{2t})^{d+1}}\frac{2t}{x}\exp\left(\frac{2t}{x}\right).
\]
Hence, we have that
\begin{align*}
\left|f(x)-\sum_{k=0}^{d}\frac{1}{k!}f^{(k)}(t)(x-t)^{k}\right|\leq & \frac{1}{d!}\left|\int_{t}^{x}\frac{(d+1)!}{(s-\frac{x}{2t})^{d+1}}\frac{2t}{x}\exp\left(\frac{2t}{x}\right)(x-s)^{d}\mathrm{d}s\right|\\
= & \frac{2t(d+1)\mathrm{e}^{\frac{2t}{x}}}{x}\int_{x}^{t}\frac{(s-x)^{d}}{(s-\frac{x}{2t})^{d+1}}\mathrm{d}s\\
\leq & \frac{4t(d+1)\mathrm{e}^{\frac{2t}{x}}}{x^2}\int_{x}^{t}\frac{(s-x)^{d}}{(s-\frac{x}{2t})^{d}}\mathrm{d}s\\
= & \frac{4(d+1)\mathrm{e}^{\frac{2t}{x}}}{x^2}\int_{x}^{t}\Big(1- \frac{(1-1/2t)x}{s-\frac{x}{2t}}\Big)^d \mathrm{d}s\\
\leq & \frac{4(d+1)\mathrm{e}^{\frac{2t}{x}}}{x^2}\int_{x}^{t}\Big(1- \frac{x}{t}\Big)^d \mathrm{d}s\\
= & \frac{4t(d+1)\mathrm{e}^{\frac{2t}{x}}}{x^2}\Big(1- \frac{x}{t}\Big)^{d+1} \\
\leq & \frac{4t(d+1)\mathrm{e}^{\frac{2t}{x}}}{x^2} e^{-x(d+1)/t}\\
\leq & 4t(d+1)\cdot \mathrm{e}^{\frac{2t-1}{x}-\log x-x(d+1)/t} f(x).
\end{align*}
In light of the above fact, we see that if $d \geq \frac{ct^2}{x^2} \log(\frac{1}{xt\error})$,
\[
\left|f(x)-\sum_{k=0}^{d}\frac{1}{k!}f^{(k)}(1)(x-1)^{k}\right|\leq \error f(x).
\]
\end{proof}

We are now ready to give proof for our main result of this section. 
\begin{proof}[Proof of \autoref{lem:JL}]
Define $g\defeq \Omega(\log^{-1}n \cdot \clb^2)$, $\me \defeq 2(u\mi-\ma)$ and note that $2\Omega(\log^{-1}n)\mi \preceq \me \preceq \frac{1}{\clb^2}\mi$. If we let $p(\me)$ be the Taylor approximation polynomial with degree $\frac{c}{g^2} \log(\frac{1}{g \error})$, for some constant $c$, then
\[
(1-\error) f(\me) \preceq p(\me) \preceq (1+\error) f(\me).
\]
Now $\me$ commutes with $f(\me)$ and therefore $p(\me)$ commutes with $f(\me)$ and by \autoref{lem:commute} we have:
\begin{equation}\label{eq:square}
(1-3\error) f(\me)^{2} \preceq p(\me)^{2} \preceq (1+3\error) f(\me)^{2}~.
\end{equation}
Next we simplify some expressions based on our previous discussions:
\[
\phideriv(u\mi-\ma)=\phideriv \left(\frac{\me}{2} \right)
= 4 \cdot \me^{-1} \exp(\me^{-1}) \exp(\me^{-1}) \me^{-1}
= 4 \cdot f(\me)^2
\]

By \autoref{eq:square} $f(\me)^2 \approx_{3\error} p(\me)^2$. Combining all this, for any matrix $\mx \succeq 0$ we have:
$$ \tr (\mx \mb^{-1/2} \phideriv(u\mi-\ma) \mb^{-1/2}) \approx_{3\error} 4 \tr (p(u\mi-\ma) \mb^{-1/2}  \mx \mb^{-1/2} p(u\mi-\ma)) $$
in the above expression we replaced $\phideriv(u\mi-\ma)$ by $p(u\mi-\ma)^2$ from our earlier discussion. Now recall $\ma=\mb^{-1/2}\hat{\mm} \mb^{-1/2}$ and $p(u\mi-\ma)$ is essentially equal to $\mb^{-1/2}q(\hat{\mm}\mb^{-1})\mb^{1/2}$ for some other polynomial $q(\cdot)$ of same degree. Further $\mb^{-1/2} p(u\mi-\ma) $ simplifies to $\mb^{-1}q(\hat{\mm}\mb^{-1})\mb^{1/2}$ and we get the following,
\begin{align*}
\tr (\mx \mb^{-1/2} \phideriv(u\mi-\ma) \mb^{-1/2}) \approx_{3\error} 4 \tr \left( \left(\mb^{-1}q(\hat{\mm}\mb^{-1})\mb^{1/2} \right)^\top \mx \left(\mb^{-1}q(\hat{\mm}\mb^{-1})\mb^{1/2}  \right) \right), 
\end{align*}
 Rewriting the right hand side of the above equation we get, 
 \[
\tr (\mx \mb^{-1/2} \phideriv(u\mi-\ma) \mb^{-1/2}) \approx_{3\error} \tr \left( \left(q(\hat{\mm}\mb^{-1})\mb^{-1}\mx^{1/2}  \right)^\top \mb \left(q(\hat{\mm}\mb^{-1})\mb^{-1}\mx^{1/2}  \right) \right)
\]
The matrix $\mb$ can now be replaced with matrix $\mc \mc^{\top}$ (Recall $\mc \mc^{\top}=\mb$) and we get,
 \[
\tr (\mx \mb^{-1/2} \phideriv(u\mi-\ma) \mb^{-1/2})  \approx_{3\error} \tr \left( \left(q(\hat{\mm}\mb^{-1})\mb^{-1}\mx^{1/2}  \right)^\top \mc \mc^{\top} \left(q(\hat{\mm}\mb^{-1})\mb^{-1}\mx^{1/2}  \right) \right)
\]
Combining all the equations and rewriting the top equation we get,
\begin{align*}
\tr (\mx \mb^{-1/2} \phideriv(u\mi-\ma) \mb^{-1/2}) \approx_{3\error} 4 \tr \left( \left(\mb^{-1}q(\hat{\mm}\mb^{-1})\mc \right)^\top \mx \left(\mb^{-1}q(\hat{\mm}\mb^{-1})\mc  \right) \right), 
\end{align*}
Now we invoke JL on RHS of previous expression to produce the following:
\begin{align*}
\tr (\mx \mb^{-1/2} \phideriv(u\mi-\ma) \mb^{-1/2}) \approx_{4\error} 4 \tr \left( \left(\mb^{-1}q(\hat{\mm}\mb^{-1})\mc \mq^\top \right)^\top \mx \left(\mb^{-1}q(\hat{\mm}\mb^{-1})\mc \mq^\top \right) \right), 
\end{align*}
where in the last two inequalities $\mq$ is a JL transformation matrix with $\tilde{O}(\frac{1}{\error^2})$ rows. We can compute $\mb^{-1}q(\hat{\mm}\mb^{-1})\mc \mq^\top$ in $\tilde{O}\left(\frac{1}{g^2\error^{O(1)}}\left(\runtime_{\mb}+\runtime_{\mb^{-1}} + \runtime_{MV} + \runtime_{SQ}\right)\right)$ time, and once we have these $\tilde{O}(\frac{1}{\error^2})$ precomputed vectors we can compute the trace product as the sum of $\tilde{O}(\frac{1}{\error^2})$ quadratic forms of $\mx$. Observe that for our purposes matrix $\mx$ takes values only from $\mathcal{M}=\{\mm'_{i}\}_{i=1}^{d}$ and we can write the whole vector $c$ in time $\rt$. The $(\clb^{-1}\ma-\ell\mi)$ case is virtually identical and also the analysis for computing vector $v$ needed inside the SDP solver is similar. Combining everything we get the proof of \autoref{lem:JL}.
\end{proof}
\section{Structure of Symmetric Inverse $M$-Matrices}
\label{sec:inv_dense}

Here we provide simple structural theorems about $M$-matrices that allow us to show that if we can solve arbitrary inverse $M$-matrix is $\tilde{O}(n^2)$ time then this suffices to solve inverse $M$-matrices in nearly linear time.

First we show that irreducible invertible $M$-matrices are dense, where recall that a matrix $\mm \in \R^{n \times n}$ is irreducible if there does not exist $S \subseteq [n]$ with $S \notin \{\emptyset, [n]\}$ such that $\mm_{ij} = 0$ for all $i \in S$ and $j \notin S$.

\begin{lemma}[\textbf{Density of Irreducible Invertible Symmetric $M$-Matrices}]
\label{lem:posminv} 
If $\mm \in \R^{n \times n}$ is an irreducible invertible symmetric $M$-matrix then $\mm_{ij} > 0$ for all $i,j \in [n]$.
\end{lemma}

\begin{proof}
Recall that $\mm$ is an invertible $M$-matrix if and only if $\mm = s \mi- \ma$ where $s > 0$, $\ma \in \R^{n \times n}_{\geq 0}$ and $\rho(\ma) < s$. In this case
\[
\left[ \mm \right]^{-1}_{ij} 
= \frac{1}{s} \left[\mi - \frac{1}{s} \ma \right]_{ij}^{-1}
= \frac{1}{s} \sum_{k = 0}^{\infty} \left[\frac{1}{s} \ma\right]_{ij}^k ~.
\]
Now, consider the undirected graph $G$ for which $\ma$ is its adjacency matrix, i.e. $\ma_{ij}$ is the weight of an edge from $i$ to $j$ whenever $\ma_{ij} \neq 0$. Now $[\ma]_{ij}^k > 0$ if and only if there is a path of length $k$ from $i$ to $j$ in $G$. However, by the assumption that $\mm$ is irreducible we have that $G$ is connected and therefore there is a path between any two vertices in the graph and the result follows.
\end{proof}

Now partitioning a symmetric matrix into irreducible components can easily be done in nearly linear time by computing connected components in the graph induced by the sparsity pattern of the matrix. Furthermore, to solve a linear system in a symmetric matrix it suffices to solve the linear system induced by each of its irreducible components independently. However, since by the previous lemma each of these components is dense we see that if we apply an $\otilde(n^2)$ time solver on each $n \times n$ block that this ultimately yields a nearly linear time algorithm for solving the inverse of $M$-matrices.

\section{M-Matrix and SDD Facts}
\label{sec:appproofs}

Here we prove omitted facts about $M$-matrices and SDD matrices from \autoref{sec:solvers}. In particular we restate and prove \autoref{lemma:mfact1} and \autoref{lemma:mfact2}

\lemmamfactone*
\begin{proof}
$\mx\mm\mx$ is trivially symmetric and therefore it suffices to show that (1)  $e_{i}\mx\mm\mx e_{j} <0$ for all  $i\neq j$ and (2) $\mx\mm\mx \vones \geq 0$.

For (1) note for all $i\in [n]$, $\mx_{ii}= e_i \mm^{-1} \vones \geq 0$ as $\mm^{-1}$ is nonnegative by \autoref{lem:posminv}  and $\mx e_{i}=\mx_{ii} e_{i}$ as  $\mx$ is diagonal matrix. Using these two equalities for all $i\neq j$ we get, $e_{i}\mx\mm\mx e_{j}=(\mx_{ii}\mx_{jj}) \cdot e_{i}\mm e_{j} \leq 0$ as  $\mm_{ij}\leq 0$ by definition of a $M$-matrix.

For (2) note $\mx \vones=\mm^{-1} \vones$ and $\mx\mm\mx \vones=\mx\mm\mm^{-1} \vones=\mx \vones=\mm^{-1} \vones \geq 0$ where again, in the last inequality we used $\mm^{-1}$ is entrywise nonnegative by \autoref{lem:posminv}.
\end{proof}

\lemmamfacttwo*

\begin{proof}
$\mb$ is clearly symmetric and therefore to show that $\mb$ is SDD with non-positive off-diagonal if suffices to show that (1) $e_{i}\mb e_{j} \leq 0$ for all $i\neq j$ and (2) $\mb \vones \geq 0$. 

Now since the claim is trivial when $\alpha = 0$ we can assume without loss of generality that $\alpha > 0$. To prove the inequality we use that by the Woodbury matrix identity it holds that
\[
\mb =
\alpha^{-1} \mi - \alpha^{-2}
\left(\ma + \alpha^{-1} \mi\right)^{-1}
=
\ma - \ma(\alpha^{-1}\mi + \ma)^{-1} \ma ~.
\]
Further, we  use that $ \ma + \alpha^{-1} \mi$ is an $M$-matrix by definition and therefore  $\left(\ma + \alpha^{-1} \mi\right)^{-1}$ is an inverse $M$-matrix that is entrywise nonnegative entries by \autoref{lem:posminv}.

Now for (1) by these two claims we have that for all $i \neq j$ it is the case that 
\[
e_{i}\mb e_{j}^\top 
= e_{i}^\top
\left(
\alpha^{-1} \mi - \alpha^{-2}
\left(\ma + \alpha^{-1} \mi\right)^{-1}
\right)
e_{j}
=
\frac{1}{\alpha^2} e_i^\top \left(\ma + \alpha^{-1} \mi\right)^{-1}  e_j \leq 0
~.
\]
For (2) we use the other Woodbury matrix equality and see that
\[
\mb \vones=\left[(\alpha^{-1}\mi + \ma) -  \ma\right](\alpha^{-1}\mi + \ma)^{-1}\ma\vones=\alpha^{-1}(\alpha^{-1}\mi + \ma)^{-1}\ma\vones \geq 0 ~.
\]
This final inequality follows from the fact that $\ma\vones\geq 0$ because $\ma$ is a SDD matrix and $(\alpha^{-1}\mi + \ma)^{-1}(\ma\vones) \geq 0$ because $(\alpha^{-1}\mi + \ma)^{-1}$ is entrywise nonnegative as discussed. 
\end{proof}

\end{document}